\PassOptionsToPackage{unicode}{hyperref}
\PassOptionsToPackage{hyphens}{url}
\PassOptionsToPackage{dvipsnames,svgnames,x11names}{xcolor}
\documentclass[11pt]{article}

\usepackage{amsmath, amssymb, amsfonts, amsthm}
\usepackage{graphicx}
\usepackage{natbib}
\usepackage{setspace}
\usepackage{geometry}
\usepackage{hyperref}
\usepackage{mathptmx}
\usepackage{booktabs}
\usepackage[table]{xcolor}
\usepackage{subfigure}
\usepackage{subcaption}
\usepackage{algorithm}
\usepackage{algorithmic}

\newtheorem{theorem}{Theorem}
\newtheorem{proposition}{Proposition}
\newtheorem{corollary}{Corollary}
\newtheorem{remark}{Remark}

\usepackage{amsmath,amssymb}
\usepackage{iftex}
\ifPDFTeX
  \usepackage[T1]{fontenc}
  \usepackage[utf8]{inputenc}
  \usepackage{textcomp} 
\else 
  \usepackage{unicode-math}
  \defaultfontfeatures{Scale=MatchLowercase}
  \defaultfontfeatures[\rmfamily]{Ligatures=TeX,Scale=1}
\fi
\usepackage{lmodern}
\ifPDFTeX\else  
\fi
\IfFileExists{upquote.sty}{\usepackage{upquote}}{}
\IfFileExists{microtype.sty}{
  \usepackage[]{microtype}
  \UseMicrotypeSet[protrusion]{basicmath} 
}{}
\makeatletter
\@ifundefined{KOMAClassName}{
  \IfFileExists{parskip.sty}{%
    \usepackage{parskip}
  }{
    \setlength{\parindent}{0pt}
    \setlength{\parskip}{6pt plus 2pt minus 1pt}}
}{
  \KOMAoptions{parskip=half}}
\makeatother
\usepackage{xcolor}
\makeatletter
\ifx\paragraph\undefined\else
  \let\oldparagraph\paragraph
  \renewcommand{\paragraph}{
    \@ifstar
      \xxxParagraphStar
      \xxxParagraphNoStar
  }
  \newcommand{\xxxParagraphStar}[1]{\oldparagraph*{#1}\mbox{}}
  \newcommand{\xxxParagraphNoStar}[1]{\oldparagraph{#1}\mbox{}}
\fi
\ifx\subparagraph\undefined\else
  \let\oldsubparagraph\subparagraph
  \renewcommand{\subparagraph}{
    \@ifstar
      \xxxSubParagraphStar
      \xxxSubParagraphNoStar
  }
  \newcommand{\xxxSubParagraphStar}[1]{\oldsubparagraph*{#1}\mbox{}}
  \newcommand{\xxxSubParagraphNoStar}[1]{\oldsubparagraph{#1}\mbox{}}
\fi
\makeatother

\usepackage{longtable,booktabs,array}
\usepackage{calc} 
\usepackage{etoolbox}
\makeatletter
\patchcmd\longtable{\par}{\if@noskipsec\mbox{}\fi\par}{}{}
\makeatother
\IfFileExists{footnotehyper.sty}{\usepackage{footnotehyper}}{\usepackage{footnote}}
\makesavenoteenv{longtable}
\usepackage{graphicx}
\makeatletter
\def\maxwidth{\ifdim\Gin@nat@width>\linewidth\linewidth\else\Gin@nat@width\fi}
\def\maxheight{\ifdim\Gin@nat@height>\textheight\textheight\else\Gin@nat@height\fi}
\makeatother
\setkeys{Gin}{width=\maxwidth,height=\maxheight,keepaspectratio}
\makeatletter
\def\fps@figure{htbp}
\makeatother

\makeatletter
\@ifpackageloaded{caption}{}{\usepackage{caption}}
\AtBeginDocument{%
\ifdefined\contentsname
  \renewcommand*\contentsname{Table of contents}
\else
  \newcommand\contentsname{Table of contents}
\fi
\ifdefined\listfigurename
  \renewcommand*\listfigurename{List of Figures}
\else
  \newcommand\listfigurename{List of Figures}
\fi
\ifdefined\listtablename
  \renewcommand*\listtablename{List of Tables}
\else
  \newcommand\listtablename{List of Tables}
\fi
\ifdefined\figurename
  \renewcommand*\figurename{Figure}
\else
  \newcommand\figurename{Figure}
\fi
\ifdefined\tablename
  \renewcommand*\tablename{Table}
\else
  \newcommand\tablename{Table}
\fi
}
\@ifpackageloaded{float}{}{\usepackage{float}}
\floatstyle{ruled}
\@ifundefined{c@chapter}{\newfloat{codelisting}{h}{lop}}{\newfloat{codelisting}{h}{lop}[chapter]}
\floatname{codelisting}{Listing}

\makeatother
\makeatletter
\@ifpackageloaded{caption}{}{\usepackage{caption}}
\@ifpackageloaded{subcaption}{}{\usepackage{subcaption}}
\makeatother

\ifLuaTeX
  \usepackage{selnolig}  
\fi
\usepackage[]{natbib}
\usepackage{bookmark}

\IfFileExists{xurl.sty}{\usepackage{xurl}}{} 
\hypersetup{
  pdftitle={Title},
  pdfauthor={Author 1; Author 2},
  pdfkeywords={3 to 6 keywords, that do not appear in the title},
  colorlinks=true,
  linkcolor={blue},
  filecolor={Maroon},
  citecolor={Blue},
  urlcolor={Blue},
  pdfcreator={LaTeX via pandoc}}

\newcommand{\anon}{1}

\begin{document}

\def\spacingset#1{\renewcommand{\baselinestretch}%
{#1}\small\normalsize} \spacingset{1}


\if1\anon
{
  \title{\bf A Complete-Data Likelihood for Epidemic Processes on Partially Observed Dynamic Networks}
  \author{Md Asaduzzaman\thanks{
    Department of Engineering, University of Staffordshire, College Road, University Quarter, Stoke on Trent ST4 2DE, United Kingdom, Email: md.asaduzzaman@staffs.ac.uk
    }\hspace{.2cm}\\
    Department of Engineering, University of Staffordshire\\
}
  \maketitle
} \fi

\if0\anon
{
  \bigskip
  \bigskip
  \bigskip
  \begin{center}
    {\LARGE\bf Title}
\end{center}
  \medskip
} \fi

\bigskip
\begin{abstract}
Inference for infectious disease transmission on dynamic contact networks is complicated by latent infection times, partially observed network evolution, measurement error in contact data, and infection originating from outside the observed population. Existing likelihood-based approaches typically address these challenges separately and often rely on restrictive assumptions such as fully observed networks, closed populations, or symptom onset as a surrogate for infection time. We develop a unified complete-data likelihood framework for epidemic processes evolving on partially observed dynamic networks. The proposed formulation represents disease progression, network evolution, and observation mechanisms as interacting continuous-time stochastic processes within a common probabilistic framework. Specifically, we couple a susceptible–exposed–infectious–removed (SEIR) epidemic process with a status-dependent dynamic contact network and explicit observation models for symptoms and contacts. The resulting framework accommodates latent incubation periods, intermittent network observation, contact measurement error, and external infection pressure while preserving a coherent likelihood structure. Our principal contribution is the derivation of a complete-data event-history likelihood for the joint epidemic–network process under partial observation. The likelihood provides a rigorous foundation for likelihood-based and Bayesian inference through data augmentation, clarifies how information from disease progression and contact dynamics jointly determines parameter estimability, and reveals a broad class of existing epidemic network models as special cases. More generally, the framework contributes to statistical inference for partially observed interacting stochastic systems on evolving networks and establishes a foundation for uncertainty-aware analysis of complex transmission processes.
\end{abstract}

\noindent%
\vfill

\newpage
\spacingset{1.2} 

\section{Introduction}
\label{sec:introduction}
Understanding how transmission processes unfold in populations requires statistical models that jointly account for disease progression and patterns of interpersonal contact \citep{buckee2021thinking}. Classical epidemic models often assume homogeneous mixing, whereby all individuals are equally likely to interact \citep{keeling2005networks}. While such assumptions can yield tractable mathematical models, they may substantially misrepresent transmission pathways when contact behaviour is heterogeneous, clustered, or evolves over time. Dynamic network models address this limitation by allowing transmission risk to depend explicitly on a time-varying contact structure, thereby linking epidemic spread to changing social, behavioural, or environmental interactions \citep{ben2010modeling, enright2018epidemics, dubey2022modeling, bu2022likelihood}.

Over the past two decades, substantial progress has been made in the development of stochastic epidemic models on networks and associated inferential methodologies \citep{allen2008introduction, kretzschmar2009mathematical, kiss2017mathematics, britton2019stochastic}. In many formulations, disease transmission is represented as a continuous-time stochastic process whose transition intensities depend on both the epidemic state and the underlying network configuration \citep{kotnis2013stochastic, gomez2023new}. Likelihood-based approaches are particularly attractive in this setting because they provide a coherent framework for parameter estimation, uncertainty quantification, model comparison, and prediction. Moreover, event-history representations establish natural connections between epidemic modelling, counting-process theory, survival analysis, and continuous-time Markov processes, allowing epidemic inference to be grounded within a broader statistical framework \citep{breto2018modeling, bu2022likelihood, whitaker2025sequential, ball2025fast}.

Despite these advances, realistic epidemic data remain only partially informative about the underlying transmission process \citep{yang2015inference, wang2018characterizing, kamkumo2025stochastic}. Infection times are rarely observed directly and must instead be inferred from symptom onset, diagnostic testing, or other imperfect proxies \citep{dawson2015epidemic, grinsztajn2021bayesian, bu2022likelihood}. Such observations introduce uncertainty due to latent incubation periods and unobserved disease progression. Contact networks are similarly difficult to observe completely \citep{eames2015six}. In practice, network information is often collected through surveys, digital contact tracing systems, proximity sensors, or administrative records, all of which may contain measurement error, missing edges, or intermittent observation \citep{grekousis2021digital}. Furthermore, many epidemic systems operate within open populations, where infection may arise from sources outside the observed network \citep{almutiry2021contact, rahnsch2024network, schneider2022epidemic}. 

These features fundamentally alter the statistical structure of the inference problem. When both disease trajectories and network evolution are only partially observed, the observed data correspond to a coarse projection of an underlying stochastic process \cite{crambes2023functional, wang2025bayesian}. Consequently, inference requires simultaneous treatment of latent disease states, latent contact histories, and imperfect observation mechanisms. Existing likelihood formulations often simplify this problem by assuming fully observed networks, known infection times, or closed populations. While such assumptions may be useful in specific applications, they restrict the scope of inference and may obscure important questions concerning parameter identifiability and uncertainty propagation \citep{bu2022likelihood, bu2025stochastic, morsomme2025exact}.

This paper develops a unified complete-data likelihood framework for epidemic processes evolving on partially observed dynamic networks. We formulate a joint continuous-time stochastic model that combines a susceptible–exposed–infectious–removed (SEIR) epidemic process with a status-dependent dynamic network and explicit observation models for symptoms and contacts. By treating both epidemic and network histories as latent processes, the proposed framework yields a complete-data event-history likelihood that accommodates latent incubation periods, intermittent network observation, contact measurement error, and external infection pressure within a single probabilistic structure.

Our contribution is methodological. Rather than introducing a new epidemic model, we derive a general likelihood framework that separates the latent process from the observation process while preserving their statistical dependence. This formulation provides a principled foundation for likelihood-based and Bayesian inference through data augmentation, clarifies how information from symptoms and contacts contributes to parameter estimation, and reveals how many existing network epidemic models arise as special cases. More broadly, the framework contributes to statistical inference for partially observed interacting stochastic systems on evolving networks, a class of problems that extends beyond epidemiology to information diffusion, behavioural contagion, and social interaction processes.

The remainder of the paper introduces the joint epidemic–network model, derives the complete-data likelihood, examines identifiability and inferential implications, and discusses computational strategies for estimation under partial observation.

\section{Background}
\label{sec:background}
We first describe the latent epidemic process through a susceptible--infectious--exposed--removed (SIER) formulation, then define the dynamic contact network on which transmission occurs, and finally specify the observation model that links the latent process to the data. This establishes a coherent probabilistic structure in which disease progression, evolving contact patterns, and imperfect observations are modeled jointly \citep{hethcote2000mathematics}.

\subsection{SIER Epidemic Model}
\label{subsec:sier_model}

Let $N$ denote the population size and index individuals by $i=1,\dots,N$. For each individual $i$, let
\[
X_i(t) \in \{S,E,I,R\}
\]
denote the disease state at time $t$, where $S$ stands for susceptible, $E$ for exposed, $I$ for infectious, and $R$ for removed. Although different applications may motivate slightly different disease compartments, the SIER structure is sufficiently rich to capture a latent incubation period while remaining analytically tractable. In particular, the exposed state allows the model to distinguish between the time of infection and the time at which an individual becomes infectious or symptomatic. Fig. \ref{fig:seir_model} demonstrtes the the diease progression on individuals.

\bigskip
\begin{figure}[htbp]
  \centering
  \includegraphics[width=0.90\textwidth]{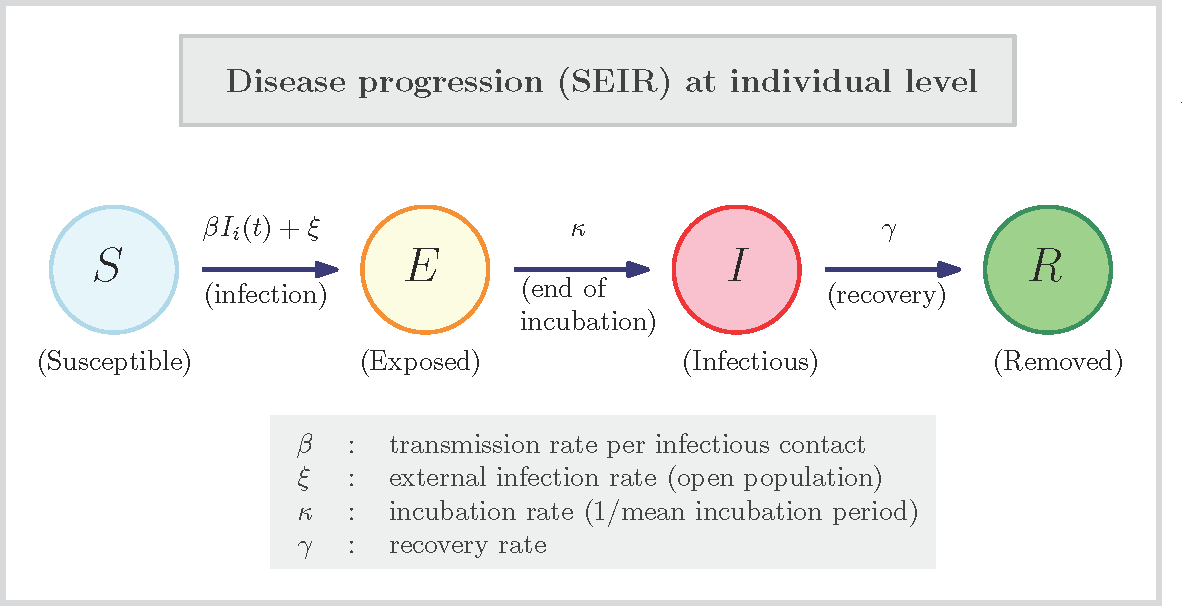}
  \caption{Structure and compartments of the SEIR model. The parameters $\beta$, $\kappa$, and $\gamma$ govern the transition rates between Susceptible ($S$), Exposed ($E$), Infectious ($I$), and Recovered ($R$) compartments, while $\xi$ represent the external infection rate}
  \label{fig:seir_model}
\end{figure}

We assume that disease progression follows a continuous-time stochastic process. A susceptible individual may become exposed through contact with infectious neighbors or through infection pressure from outside the observed network. Once exposed, an individual progresses to the infectious state after a random waiting time, and an infectious individual eventually transitions to the removed state. This gives the transition sequence
\[
S \to E \to I \to R.
\]
Let $\beta$ denote the within-network transmission rate, $\xi$ the external infection rate, $\kappa$ the incubation progression rate, and $\gamma$ the removal rate. Then the transition hazards may be written as
\[
\lambda_i^{SE}(t)=\beta I_i(t)+\xi,\qquad
\lambda_i^{EI}(t)=\kappa,\qquad
\lambda_i^{IR}(t)=\gamma,
\]
where $I_i(t)$ is the number of infectious contacts of individual $i$ at time $t$.

The key feature of this formulation is that infection risk is not homogeneous across individuals. Rather, it depends on the current contact structure and the infection status of nearby individuals. This allows the model to represent transmission clustering, repeated exposure, and heterogeneity in connectivity. The latent epidemic process is therefore not defined only through marginal disease-state transitions, but through its coupling to the evolving social or contact network. Such coupling is essential when the transmission mechanism is mediated by interactions that are themselves time varying.

The SIER model is especially useful in settings where infection is not observed directly. In many practical applications, only symptom reports, test results, or diagnosis times are available. These observations are often delayed relative to infection, so the latent exposed state captures an important source of uncertainty. From a statistical perspective, this latent structure prevents the analyst from equating symptom onset with infection time and provides a more realistic model for inference.

\subsection{Dynamic Contact Network}
\label{subsec:dynamic_network}

Transmission in the proposed framework occurs over a dynamic contact network. Let
\[
A_{ij}(t)\in\{0,1\}, \qquad 1\le i<j\le N,
\]
denote the contact indicator between individuals $i$ and $j$ at time $t$, where $A_{ij}(t)=1$ indicates that a contact is present and $A_{ij}(t)=0$ indicates no contact. We assume symmetry, so that $A_{ij}(t)=A_{ji}(t)$, and no self-contacts, so that $A_{ii}(t)=0$.

Unlike static-network models, the dynamic network evolves over time and may depend on the disease states of the incident nodes. This is important because contact behavior often changes in response to infection status, symptoms, awareness, or intervention. For example, infectious individuals may reduce social interactions, susceptible individuals may avoid contact with known cases, and removed individuals may no longer contribute to transmission. To accommodate such behavior, we model the contact network as a continuous-time Markov process with state-dependent transition intensities as shown in Fig.\ \ref{fig:dcn_model}.

\bigskip
\begin{figure}[htbp]
  \centering
  \includegraphics[width=0.95\textwidth]{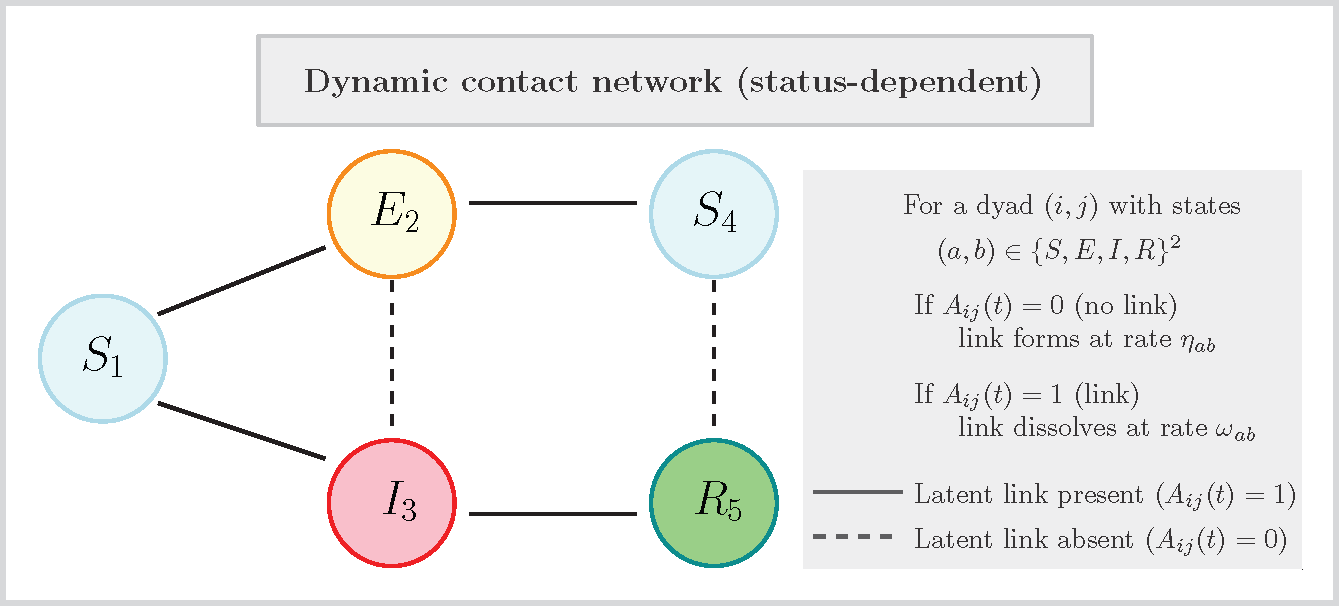}
  \caption{Dynamic contact network}
  \label{fig:dcn_model}
\end{figure}

Specifically, for each dyad $(i,j)$, the edge may form or dissolve according to the current pair of disease states $(X_i(t),X_j(t))$. If the edge is absent, it forms at rate $\eta_{ab}$ when the incident states are $(a,b)$; if the edge is present, it dissolves at rate $\tau_{ab}$. Thus,
\[
A_{ij}(t): 0 \to 1 \quad \text{at rate } \eta_{X_i(t),X_j(t)},
\]
and
\[
A_{ij}(t): 1 \to 0 \quad \text{at rate } \tau_{X_i(t),X_j(t)}.
\]
We impose symmetry constraints
\[
\eta_{ab}=\eta_{ba}, \qquad \tau_{ab}=\tau_{ba},
\]
which reflect the undirected nature of the contact network.

This dynamic network representation is deliberately flexible. It allows contacts to be persistent or transient, random or status dependent, and sparse or dense depending on the parameter values. It also permits the disease process to feed back into the network process through state-dependent contact behavior. Such feedback is often present in real epidemics and can materially affect both the speed and the structure of spread. Ignoring this dependence may lead to misleading inference, especially when contact tracing, behavioral change, or risk avoidance are present.

From a modeling standpoint, the network process and the epidemic process should be viewed as a single coupled system rather than separate components. The contact network determines who can infect whom, while the epidemic states may alter how contacts form and dissolve. This mutual dependence is what makes the joint model both scientifically realistic and statistically challenging. It also motivates the need for a complete-data likelihood in which both latent histories are treated as unknown and estimated jointly from the data.

\subsection{Observation Model}
\label{subsec:observation_model}

In practice, neither the epidemic states nor the contact network are observed continuously or without error. The data typically consist of discrete-time measurements that may be incomplete, noisy, or misclassified. The observation model describes how these recorded data are generated from the latent epidemic--network process (as shown in Fig.\ \ref{fig:om_model}).

\begin{figure}[htbp]
  \centering
  \includegraphics[width=0.80\textwidth]{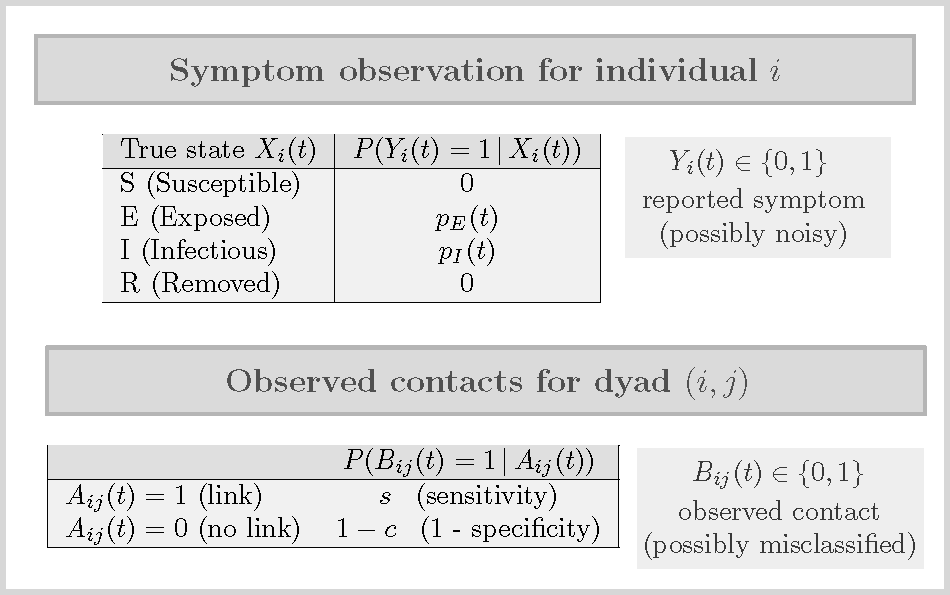}
  \caption{Symtom observation for individual $i$ and observed contacts for dyads $(i,j)$}
  \label{fig:om_model}
\end{figure}

Suppose observations are available at times
\[
0<t_1<\cdots<t_M\le T.
\]
At each observation time, we may observe symptom reports, diagnostic outcomes, or other disease-related indicators. Let
\[
Y_i(t_m)
\]
denote the observed disease measurement for individual $i$ at time $t_m$. In the simplest binary form, $Y_i(t_m)$ may indicate whether symptoms are reported or whether a test is positive. Conditional on the latent state $X_i(t_m)$, the observation distribution is modeled through a misclassification mechanism. For example,
\[
\Pr\{Y_i(t_m)=1\mid X_i(t_m)=E\}=p_E,\qquad
\Pr\{Y_i(t_m)=1\mid X_i(t_m)=I\}=p_I,
\]
with corresponding false-positive or false-negative probabilities allowed for the susceptible and removed states if needed.

Likewise, contact observations are often imperfect. Let
\[
B_{ij}(t_m)\in\{0,1\}
\]
denote the observed contact indicator for dyad $(i,j)$ at time $t_m$. This observed network may be obtained from surveys, wearable sensors, contact-tracing logs, or administrative records, each of which may miss true contacts or record spurious ones. A simple but useful error model is
\[
\Pr\{B_{ij}(t_m)=1\mid A_{ij}(t_m)=1\}=s,
\qquad
\Pr\{B_{ij}(t_m)=1\mid A_{ij}(t_m)=0\}=1-c,
\]
where $s$ is sensitivity and $c$ is specificity.

This observation layer is essential because the latent process is rarely directly visible. The analyst often observes only a noisy projection of the true epidemic and network histories. The observation model therefore links the scientific process to the available data, while also accounting for measurement error and intermittent missingness. Conditional on the latent trajectory, the symptom and contact observations are assumed to be independent across individuals, dyads, and observation times, which yields a tractable likelihood decomposition.

The combination of latent epidemic dynamics, evolving contact structure, and imperfect observation produces the central inferential challenge of the paper. The complete-data likelihood will exploit this structure by treating the latent process as missing data and integrating the observation model into a unified statistical framework. This formulation provides a principled basis for likelihood-based estimation, Bayesian inference, and identifiability analysis.

\section{Mathematical Formulation}
\label{sec:math_formulation}

\subsection{State Space and Joint Process}

Let $V=\{1,\dots,N\}$ denote the set of individuals in the population. For each $i\in V$, let
\[
X_i(t)\in\mathcal{S}=\{S,E,I,R\}
\]
denote the disease state of individual $i$ at time $t$, and for each unordered pair $(i,j)$ with $1\le i<j\le N$, let
\[
A_{ij}(t)\in\{0,1\}
\]
denote the status of the contact edge between $i$ and $j$ at time $t$. Define the full latent state as
\[
Z(t)=\Bigl(X_1(t),\dots,X_N(t),\{A_{ij}(t):1\le i<j\le N\}\Bigr).
\]
We assume that $\{Z(t):0\le t\le T\}$ is a continuous-time Markov chain on a finite state space $\mathcal{Z}$.

For any latent state $z\in\mathcal{Z}$, write $x_i(z)$ for the disease state of node $i$ and $a_{ij}(z)$ for the dyad state. The process is governed by a generator
\[
\mathcal{Q}_\theta = \bigl(q_\theta(z,z')\bigr)_{z,z'\in\mathcal{Z}},
\]
with parameters $\theta=(\theta_{\mathrm{epi}},\theta_{\mathrm{net}},\theta_{\mathrm{obs}})$.

\subsection{Epidemic Transition Intensities}

Define the infectious neighborhood count for node $i$ under latent state $z$ by
\[
I_i(z)=\sum_{j\ne i} a_{ij}(z)\,\mathbf{1}\{x_j(z)=I\}.
\]
The epidemic intensities are specified as
\[
\lambda_i^{SE}(z)=\beta I_i(z)+\xi,
\qquad
\lambda_i^{EI}(z)=\kappa,
\qquad
\lambda_i^{IR}(z)=\gamma,
\]
where $\beta>0$ is the within-network transmission rate, $\xi\ge 0$ is the external infection rate, $\kappa>0$ is the incubation progression rate, and $\gamma>0$ is the recovery rate.

Thus, if $z^{i,S\to E}$ denotes the latent state obtained from $z$ by changing $x_i(z)=S$ to $E$, then
\[
q_\theta(z,z^{i,S\to E})=\lambda_i^{SE}(z).
\]
Similarly,
\[
q_\theta(z,z^{i,E\to I})=\kappa,\qquad
q_\theta(z,z^{i,I\to R})=\gamma.
\]

\subsection{Dynamic Network Intensities}

Let $\eta_{ab}$ and $\tau_{ab}$ denote the formation and dissolution parameters for a dyad whose endpoint disease states are $(a,b)\in\mathcal{S}^2$. We impose symmetry:
\[
\eta_{ab}=\eta_{ba},\qquad \tau_{ab}=\tau_{ba}.
\]
If $a_{ij}(z)=0$ and the endpoint states are $(a,b)$, then the formation intensity is
\[
q_\theta(z,z^{ij,+})=\eta_{ab},
\]
where $z^{ij,+}$ denotes the state obtained by setting $a_{ij}=1$. If $a_{ij}(z)=1$, then the dissolution intensity is
\[
q_\theta(z,z^{ij,-})=\tau_{ab},
\]
where $z^{ij,-}$ is the state obtained by setting $a_{ij}=0$.\\
The network parameter vector is
\[
\theta_{\mathrm{net}}=\{\eta_{ab},\tau_{ab}:(a,b)\in\mathcal{S}^2, \, a\le b\}.
\]

\subsection{Total Exit Rate}

For a given latent state $z$, the total hazard is
\[
\Lambda_\theta(z)=\sum_{z'\ne z} q_\theta(z,z').
\]
Equivalently, this may be decomposed into epidemic and network components:
\[
\Lambda_\theta(z)=\Lambda_{\mathrm{epi}}(z)+\Lambda_{\mathrm{net}}(z),
\]
where
\[
\Lambda_{\mathrm{epi}}(z)=\sum_{i:x_i(z)=S}\lambda_i^{SE}(z)+N_E(z)\kappa+N_I(z)\gamma,
\]
with $N_E(z)$ and $N_I(z)$ denoting the number of exposed and infectious individuals, and
\[
\Lambda_{\mathrm{net}}(z)=\sum_{1\le i<j\le N}
\Bigl\{
\mathbf{1}\{a_{ij}(z)=0\}\eta_{x_i(z)x_j(z)}
+
\mathbf{1}\{a_{ij}(z)=1\}\tau_{x_i(z)x_j(z)}
\Bigr\}.
\]

\subsection{Observed Data Model}

Suppose observations are collected at times $0<t_1<\cdots<t_M\le T$. Let
\[
Y_i(t_m)\in\mathcal{Y}
\]
denote the observed disease-related measurement for individual $i$ at time $t_m$, and let
\[
B_{ij}(t_m)\in\mathcal{B}
\]
denote the observed contact indicator for dyad $(i,j)$ at time $t_m$.\\
We specify conditional observation kernels
\[
g_\theta\bigl(Y_i(t_m)\mid X_i(t_m)\bigr),
\qquad
h_\theta\bigl(B_{ij}(t_m)\mid A_{ij}(t_m)\bigr),
\]
with conditional independence across individuals, dyads, and observation times given the latent process. For binary symptom data and binary contact data, a common choice is
\[
\Pr\{Y_i(t_m)=1\mid X_i(t_m)=E\}=p_E,\qquad
\Pr\{Y_i(t_m)=1\mid X_i(t_m)=I\}=p_I,
\]
and
\[
\Pr\{B_{ij}(t_m)=1\mid A_{ij}(t_m)=1\}=s,\qquad
\Pr\{B_{ij}(t_m)=1\mid A_{ij}(t_m)=0\}=1-c,
\]
where $s$ and $c$ denote sensitivity and specificity.

\section{Complete-Data Likelihood}
\label{sec:complete_likelihood}

\subsection{Event-History Representation}

Let $0<\tau_1<\cdots<\tau_K\le T$ denote the ordered latent event times, and let $e_k$ denote the event type at time $\tau_k$. The latent path likelihood under the CTMC representation is
\[
p_\theta\bigl(Z_{[0,T]}\bigr)
=
\prod_{k=1}^K q_\theta\!\bigl(Z(\tau_k^-),Z(\tau_k)\bigr)
\exp\left\{
-\int_0^T \Lambda_\theta\bigl(Z(u)\bigr)\,du
\right\}.
\]

Equivalently, if the path is represented through sufficient event counts and integrated risks, one may write
\[
\log p_\theta\bigl(Z_{[0,T]}\bigr)
=
\sum_{k=1}^K \log q_\theta\!\bigl(Z(\tau_k^-),Z(\tau_k)\bigr)
-
\int_0^T \Lambda_\theta\bigl(Z(u)\bigr)\,du.
\]

\subsection{Observation Likelihood}

Conditional on the latent path, the observation likelihood is
\[
p_\theta(Y,B\mid Z_{[0,T]})
=
\prod_{m=1}^M \prod_{i=1}^N g_\theta\bigl(Y_i(t_m)\mid X_i(t_m)\bigr)
\prod_{m=1}^M \prod_{1\le i<j\le N} h_\theta\bigl(B_{ij}(t_m)\mid A_{ij}(t_m)\bigr).
\]

\subsection{Complete-Data Likelihood}

The complete-data likelihood is therefore
\[
L(\theta;Z_{[0,T]},Y,B)
=
p_\theta\bigl(Z_{[0,T]}\bigr)\,
p_\theta(Y,B\mid Z_{[0,T]}).
\]
Thus,
\begin{align*}
\log L(\theta;Z_{[0,T]},Y,B) 
&= \sum_{k=1}^K \log q_\theta\!\bigl(Z(\tau_k^-),Z(\tau_k)\bigr) - \int_0^T \Lambda_\theta\bigl(Z(u)\bigr)\,du \\
&\quad\quad\quad\quad\quad\quad + \sum_{m=1}^M \sum_{i=1}^N \log g_\theta\bigl(Y_i(t_m)\mid X_i(t_m)\bigr) \\
&\quad\quad\quad\quad\quad\quad\quad\quad\quad\quad\quad + \sum_{m=1}^M \sum_{1\le i<j\le N} \log h_\theta\bigl(B_{ij}(t_m)\mid A_{ij}(t_m)\bigr).
\end{align*}

\subsection{Proposition 1: Markov Property}
\begin{proposition}
The joint latent process $Z(t)=(X(t),A(t))$ is a continuous-time Markov chain with respect to its natural filtration, provided that all transition intensities depend on the current latent state only through $Z(t)$.
\end{proposition}

\begin{proof}
For any future time increment, the probability of a transition over a small interval depends only on the present state through the generator $\mathcal{Q}_\theta$. Since all epidemic and network intensities are functions of $Z(t)$ alone, the conditional distribution of the future given the past depends only on $Z(t)$. Hence the process is Markov.
\end{proof}

\subsection{Proposition 2: Factorization of the Complete Likelihood}

\begin{proposition}
Under conditional independence of observation kernels given the latent path, the complete-data likelihood factorizes into latent-process and observation components as
\[
L(\theta;Z_{[0,T]},Y,B)
=
p_\theta(Z_{[0,T]})\,p_\theta(Y,B\mid Z_{[0,T]}).
\]
\end{proposition}

\begin{proof}
This follows directly from the joint law of the latent path and observations and the assumed conditional independence structure. Since the observation model depends on the latent trajectory only through the current values of $(X_i(t_m),A_{ij}(t_m))$, the stated factorization is immediate.
\end{proof}

\subsection{Proposition 3: Observed-Data Intractability}

\begin{proposition}
The observed-data likelihood
\[
p_\theta(Y,B)=\int p_\theta(Z_{[0,T]},Y,B)\,dZ_{[0,T]}
\]
is intractable in general for moderate to large $N$.
\end{proposition}

\begin{proof}
The integral ranges over all latent epidemic histories, all latent network trajectories, and all compatible event times. Since the number of possible latent paths grows combinatorially with $N$ and with the number of unobserved events, analytic marginalization is infeasible except in highly simplified special cases.
\end{proof}

\subsection{Corollary 1: Basis for Data Augmentation}

\begin{corollary}
Because the complete-data likelihood is available in closed event-history form, inference may be implemented by augmenting the latent event times, latent states, and unobserved edges, and then alternating between path updates and parameter updates.
\end{corollary}

\subsection{Identifiability Considerations}

The parameters $(\beta,\xi)$ may be weakly separable when infection times are not directly observed and contact data are sparse. Likewise, the pairwise network parameters $(\eta_{ab},\tau_{ab})$ may not be identifiable without sufficient temporal resolution in the observed network. These limitations are intrinsic to the partial observation regime and should be characterized through simulation, sensitivity analysis, or formal identifiability arguments.

\subsection{Special Cases}

The framework contains several important special cases:
\begin{itemize}
\item If the network is fully observed, the model reduces to an epidemic process on a known dynamic contact graph.
\item If the network is static, all network transition intensities vanish and the model becomes a SEIR epidemic model on a fixed graph.
\item If infection times are directly observed, the latent epidemic component simplifies substantially and inference for $(\beta,\xi,\kappa,\gamma)$ becomes easier.
\item If $\xi=0$, the population is closed with respect to infection importation.
\end{itemize}

These reductions are useful both for interpretation and for validating the proposed methodology against simpler benchmark settings.

\section{Inferential Implications}
\label{sec:inf_imp}
\subsection{Derivation of the Complete-Data MLE}
\label{subsec:complete_data_mle}

To clarify the inferential structure of the proposed model, we first derive the maximum likelihood estimators under the complete-data formulation, that is, conditional on the full latent epidemic and network trajectory. Although these estimators are not directly available from the observed data alone, they provide the natural target for likelihood-based and data-augmentation methods.

Let
\[
Z_{[0,T]}=\{X(t),A(t):0\le t\le T\}
\]
denote the complete latent trajectory over the observation window. Recall that the complete-data likelihood factorizes as
\[
L(\theta;Z_{[0,T]},Y,B)
=
p_\theta\!\left(Z_{[0,T]}\right)
\,p_\theta\!\left(Y,B\mid Z_{[0,T]}\right),
\]
where $\theta=(\theta_{\mathrm{epi}},\theta_{\mathrm{net}},\theta_{\mathrm{obs}})$ collects all epidemic, network, and observation parameters. Because the latent process is a continuous-time Markov chain, its density admits the standard event-history representation
\[
p_\theta\!\left(Z_{[0,T]}\right)
=
\left\{
\prod_{k=1}^{K} q_\theta\!\left(z_{k^-},z_k\right)
\right\}
\exp\!\left\{
-\int_0^T \Lambda_\theta\!\left(Z(u)\right)\,du
\right\},
\]
where $z_{k^-}$ and $z_k$ denote the latent states immediately before and after the $k$th event time, and $\Lambda_\theta(z)$ is the total exit rate from state $z$.

\begin{theorem}[Complete-data MLE]
\label{thm:complete_data_mle}
Under the complete-data formulation, the log-likelihood separates into epidemic, network, and observation components, each of which yields a closed-form maximum likelihood estimator of event-count-over-exposure form.
\end{theorem}

\begin{proof}
We derive the estimators block by block.
\end{proof}

\paragraph{Epidemic parameters.}
Let $N_{SE}^{\mathrm{int}}$ denote the number of internal $S\to E$ events, $N_{SE}^{\mathrm{ext}}$ the number of external $S\to E$ events, $N_{EI}$ the number of $E\to I$ events, and $N_{IR}$ the number of $I\to R$ events. Define the integrated exposure processes
\[
U_{SI}=\int_0^T S_I(u)\,du,\qquad
U_S=\int_0^T S(u)\,du,\qquad
U_E=\int_0^T E(u)\,du,\qquad
U_I=\int_0^T I(u)\,du,
\]
where $S_I(u)$ denotes the total susceptible--infectious contact exposure at time $u$.

Under the epidemic hazard specification
\[
\lambda_i^{SE}(u)=\beta I_i(u)+\xi,\qquad
\lambda^{EI}(u)=\kappa,\qquad
\lambda^{IR}(u)=\gamma,
\]
the epidemic contribution to the complete-data log-likelihood is
\[
\ell_{\mathrm{epi}}(\beta,\xi,\kappa,\gamma)
=
N_{SE}^{\mathrm{int}}\log\beta
-
\beta U_{SI}
+
N_{SE}^{\mathrm{ext}}\log\xi
-
\xi U_S
+
N_{EI}\log\kappa
-
\kappa U_E
+
N_{IR}\log\gamma
-
\gamma U_I,
\]
up to an additive constant. Differentiating with respect to each parameter yields
\[
\frac{\partial \ell_{\mathrm{epi}}}{\partial \beta}
=
\frac{N_{SE}^{\mathrm{int}}}{\beta}-U_{SI},
\qquad
\frac{\partial \ell_{\mathrm{epi}}}{\partial \xi}
=
\frac{N_{SE}^{\mathrm{ext}}}{\xi}-U_S,
\]
\[
\frac{\partial \ell_{\mathrm{epi}}}{\partial \kappa}
=
\frac{N_{EI}}{\kappa}-U_E,
\qquad
\frac{\partial \ell_{\mathrm{epi}}}{\partial \gamma}
=
\frac{N_{IR}}{\gamma}-U_I.
\]
Setting these derivatives to zero gives the closed-form estimators
\[
\widehat{\beta}
=
\frac{N_{SE}^{\mathrm{int}}}{U_{SI}},
\qquad
\widehat{\xi}
=
\frac{N_{SE}^{\mathrm{ext}}}{U_S},
\qquad
\widehat{\kappa}
=
\frac{N_{EI}}{U_E},
\qquad
\widehat{\gamma}
=
\frac{N_{IR}}{U_I}.
\]
The second derivatives are strictly negative whenever the corresponding event counts are positive, so these critical points are maxima.

\paragraph{Network parameters.}
For each unordered disease-state pair $(a,b)\in\mathcal{S}^2$ with $a\le b$, let $N_{ab}^{0\to1}$ denote the number of edge-formation events among dyads whose endpoint states are $(a,b)$, and let $N_{ab}^{1\to0}$ denote the number of edge-dissolution events among dyads of the same type. Define the integrated at-risk processes
\[
V_{ab}^{0}=\int_0^T Y_{ab}^{0}(u)\,du,
\qquad
V_{ab}^{1}=\int_0^T Y_{ab}^{1}(u)\,du,
\]
where $Y_{ab}^{0}(u)$ is the number of absent dyads of type $(a,b)$ at time $u$, and $Y_{ab}^{1}(u)$ is the number of present dyads of type $(a,b)$ at time $u$.

Under the status-dependent network intensities
\[
q_\theta(z,z^{ij,+})=\eta_{ab},
\qquad
q_\theta(z,z^{ij,-})=\tau_{ab},
\]
the network contribution to the complete-data log-likelihood is
\[
\ell_{\mathrm{net}}(\eta,\tau)
=
\sum_{a\le b}
\left[
N_{ab}^{0\to1}\log\eta_{ab}
-
\eta_{ab}V_{ab}^{0}
+
N_{ab}^{1\to0}\log\tau_{ab}
-
\tau_{ab}V_{ab}^{1}
\right].
\]
Differentiation gives
\[
\frac{\partial \ell_{\mathrm{net}}}{\partial \eta_{ab}}
=
\frac{N_{ab}^{0\to1}}{\eta_{ab}}-V_{ab}^{0},
\qquad
\frac{\partial \ell_{\mathrm{net}}}{\partial \tau_{ab}}
=
\frac{N_{ab}^{1\to0}}{\tau_{ab}}-V_{ab}^{1}.
\]
Hence the maximum likelihood estimators are
\[
\widehat{\eta}_{ab}
=
\frac{N_{ab}^{0\to1}}{V_{ab}^{0}},
\qquad
\widehat{\tau}_{ab}
=
\frac{N_{ab}^{1\to0}}{V_{ab}^{1}}.
\]

\paragraph{Observation parameters.}
Let $Y_i(t_m)\in\{0,1\}$ denote the symptom observation for individual $i$ at time $t_m$. Suppose
\[
\Pr\{Y_i(t_m)=1\mid X_i(t_m)=E\}=p_E,
\qquad
\Pr\{Y_i(t_m)=1\mid X_i(t_m)=I\}=p_I.
\]
Define
\[
M_E=\sum_{m=1}^M\sum_{i=1}^N \mathbf{1}\{X_i(t_m)=E\},
\qquad
R_E=\sum_{m=1}^M\sum_{i=1}^N \mathbf{1}\{Y_i(t_m)=1,X_i(t_m)=E\},
\]
and similarly $M_I$ and $R_I$. Then
\[
\ell_{\mathrm{sym}}(p_E,p_I)
=
R_E\log p_E+(M_E-R_E)\log(1-p_E)
+
R_I\log p_I+(M_I-R_I)\log(1-p_I),
\]
which yields
\[
\widehat{p}_E=\frac{R_E}{M_E},
\qquad
\widehat{p}_I=\frac{R_I}{M_I}.
\]

For the contact observation model, let $B_{ij}(t_m)\in\{0,1\}$ denote the observed contact indicator and assume
\[
\Pr\{B_{ij}(t_m)=1\mid A_{ij}(t_m)=1\}=s,
\qquad
\Pr\{B_{ij}(t_m)=1\mid A_{ij}(t_m)=0\}=1-c.
\]
Define
\[
M_1=\sum_{m=1}^M\sum_{1\le i<j\le N}\mathbf{1}\{A_{ij}(t_m)=1\},
\qquad
R_1=\sum_{m=1}^M\sum_{1\le i<j\le N}\mathbf{1}\{B_{ij}(t_m)=1,A_{ij}(t_m)=1\},
\]
and
\[
M_0=\sum_{m=1}^M\sum_{1\le i<j\le N}\mathbf{1}\{A_{ij}(t_m)=0\},
\qquad
R_0=\sum_{m=1}^M\sum_{1\le i<j\le N}\mathbf{1}\{B_{ij}(t_m)=0,A_{ij}(t_m)=0\}.
\]
The corresponding log-likelihood is
\[
\ell_{\mathrm{cont}}(s,c)
=
R_1\log s+(M_1-R_1)\log(1-s)
+
R_0\log c+(M_0-R_0)\log(1-c),
\]
and the MLEs are
\[
\widehat{s}=\frac{R_1}{M_1},
\qquad
\widehat{c}=\frac{R_0}{M_0}.
\]
Combining the blocks above, the complete-data maximum likelihood estimator is
\[
\widehat{\theta}
=
\left(
\widehat{\beta},\widehat{\xi},\widehat{\kappa},\widehat{\gamma},
\{\widehat{\eta}_{ab},\widehat{\tau}_{ab}\}_{a\le b},
\widehat{p}_E,\widehat{p}_I,\widehat{s},\widehat{c}
\right).
\]
These estimators have the familiar event-count-over-exposure form because, conditional on the latent path, each transition mechanism behaves as a Poisson-type counting process with piecewise constant risk set.

\paragraph{Remark.}
The expressions above are complete-data MLEs. In the observed-data problem, the latent epidemic history and network trajectory are not known, so these formulas are not directly evaluable. They nevertheless provide the analytic target for EM, Monte Carlo EM, and Bayesian data-augmentation approaches, and they clarify the statistical meaning of each parameter block in the proposed model.

The complete-data formulation provides a unified foundation for both likelihood-based and Bayesian inference. In a likelihood framework, one may maximize the observed or augmented likelihood using Monte Carlo EM or related simulation-based optimization methods. In a Bayesian framework, the same complete-data structure supports posterior sampling over parameters and latent trajectories.

The principal inferential challenge is that epidemic and network parameters are often entangled under partial observation. Richer observation schemes improve identifiability by separately informing transmission, incubation, recovery, and contact dynamics. Accordingly, the model should be assessed under multiple observation designs to determine which components of the latent process are estimable from the available data.

\subsection{Asymptotic Normality of the Complete-Data MLE}
\label{subsec:asymptotic_normality}

The complete-data maximum likelihood estimators derived in Section~\ref{sec:inf_imp} have the familiar event-count-over-exposure form. This structure suggests a standard counting-process asymptotic theory: as the amount of latent event information increases, the estimators are consistent and asymptotically normal with covariance determined by the inverse Fisher information. We state the result in a form adapted to the epidemic, network, and observation blocks of the model.

\begin{theorem}[Asymptotic normality of the complete-data MLE]
\label{thm:asymp_normality}
Let $\widehat{\theta}$ denote the complete-data maximum likelihood estimator based on the full latent trajectory
\[
Z_{[0,T]}=\{X(t),A(t):0\le t\le T\},
\]
and assume that the latent epidemic and network processes satisfy standard regularity conditions for multitype counting processes: 
(i) the true parameter $\theta_0$ lies in the interior of the parameter space;
(ii) the relevant exposure processes are finite and strictly positive with probability tending to one;
(iii) the intensities are correctly specified and continuously differentiable in $\theta$;
and (iv) a law of large numbers and martingale central limit theorem apply to the associated counting processes.

Then, as the effective sample size $n\to\infty$,
\[
\sqrt{n}\,(\widehat{\theta}-\theta_0)\ \xrightarrow{d}\ N\!\left(0,\; \mathcal{I}(\theta_0)^{-1}\right),
\]
where $\mathcal{I}(\theta_0)$ is the Fisher information matrix for the complete-data likelihood. Equivalently, the observed information at the MLE satisfies
\[
\mathcal{J}_n(\widehat{\theta})/n \ \xrightarrow{p}\ \mathcal{I}(\theta_0),
\]
so that the asymptotic covariance may be estimated by $\mathcal{J}_n(\widehat{\theta})^{-1}$.
\end{theorem}

\begin{proof}
The complete-data log-likelihood decomposes into additive epidemic, network, and observation components. Each block is of the generic counting-process form
\[
\ell(\psi)=\sum_{r=1}^{R_n}\log \lambda_\psi(t_r)-\int_0^T \Lambda_\psi(t)\,dt,
\]
where $\psi$ denotes a subvector of parameters, $\{t_r\}$ are event times, and $\Lambda_\psi(t)$ is the total hazard. For such models, the score can be written as a martingale integral,
\[
U_n(\psi)=\frac{\partial \ell(\psi)}{\partial \psi}
=\int_0^T \dot{\log \lambda_\psi}(t)\,dM(t),
\]
where $M(t)$ is the associated counting-process martingale. Under the stated regularity conditions, the compensator term obeys a law of large numbers, while the martingale term satisfies a central limit theorem. Hence,
\[
n^{-1/2}U_n(\theta_0)\ \xrightarrow{d}\ N(0,\mathcal{I}(\theta_0)).
\]
A Taylor expansion of the score equation $U_n(\widehat{\theta})=0$ around $\theta_0$ yields
\[
0=U_n(\theta_0)+\mathcal{J}_n(\tilde{\theta})(\widehat{\theta}-\theta_0),
\]
for some $\tilde{\theta}$ between $\widehat{\theta}$ and $\theta_0$. Dividing by $\sqrt{n}$ and using consistency of $\widehat{\theta}$ together with
\[
\mathcal{J}_n(\tilde{\theta})/n \xrightarrow{p} \mathcal{I}(\theta_0),
\]
gives
\[
\sqrt{n}(\widehat{\theta}-\theta_0)
=
-\left\{\mathcal{J}_n(\tilde{\theta})/n\right\}^{-1}
\left\{n^{-1/2}U_n(\theta_0)\right\}
\xrightarrow{d} N\!\left(0,\mathcal{I}(\theta_0)^{-1}\right).
\]
This establishes asymptotic normality.
\end{proof}

\begin{remark}[Blockwise information]
Because the complete-data likelihood separates into epidemic, network, and observation blocks, the Fisher information is typically close to block diagonal. In particular, the epidemic block for $(\beta,\xi,\kappa,\gamma)$ is driven by event counts and integrated exposure times, while the network block for $(\eta_{ab},\tau_{ab})$ is driven by edge formation and dissolution counts. The observation block for $(p_E,p_I,s,c)$ is determined by Bernoulli/binomial terms. This separation yields transparent standard-error formulas and simplifies simulation-based coverage assessment.
\end{remark}

\begin{corollary}[Asymptotic standard errors]
\label{cor:asymp_se}
Under the conditions of Theorem~\ref{thm:asymp_normality},
\[
\widehat{\theta}_j \pm z_{1-\alpha/2}\sqrt{\left\{\mathcal{J}_n(\widehat{\theta})^{-1}\right\}_{jj}}
\]
is an asymptotically valid $(1-\alpha)$ confidence interval for each scalar component $\theta_j$.
\end{corollary}

\subsection{Information Matrix and Curvature}
\label{subsec:information_curvature}

A key advantage of the complete-data formulation is that the log-likelihood is a sum of simple concave components, each corresponding to one parameter block. This makes the curvature of the likelihood analytically tractable and yields closed-form expressions for the observed and expected information matrices. These expressions provide the basis for standard errors, Wald intervals, and asymptotic coverage calculations in the simulation study.

\paragraph{General form.}
Let $\ell_c(\theta)$ denote the complete-data log-likelihood, where
\[
\theta = (\theta_{\mathrm{epi}}, \theta_{\mathrm{net}}, \theta_{\mathrm{obs}})
\]
collects the epidemic, network, and observation parameters. The observed information matrix is defined as
\[
\mathcal{J}(\theta)
=
-\frac{\partial^2 \ell_c(\theta)}{\partial \theta \,\partial \theta^\top},
\]
and the expected Fisher information is
\[
\mathcal{I}(\theta)
=
\mathbb{E}_{\theta}\!\left[
-\frac{\partial^2 \ell_c(\theta)}{\partial \theta \,\partial \theta^\top}
\right].
\]
Because the complete-data likelihood factorizes by parameter block, both matrices are nearly block diagonal, with one block for epidemic parameters, one for network parameters, and one for observation parameters.

\paragraph{Epidemic block.}
Recall the complete-data epidemic log-likelihood
\[
\ell_{\mathrm{epi}}(\beta,\xi,\kappa,\gamma)
=
N_{SE}^{\mathrm{int}}\log \beta - \beta U_{SI}
+
N_{SE}^{\mathrm{ext}}\log \xi - \xi U_S
+
N_{EI}\log \kappa - \kappa U_E
+
N_{IR}\log \gamma - \gamma U_I
+ C.
\]
The Hessian with respect to $(\beta,\xi,\kappa,\gamma)$ is diagonal:
\[
-\frac{\partial^2 \ell_{\mathrm{epi}}}{\partial (\beta,\xi,\kappa,\gamma)\,\partial (\beta,\xi,\kappa,\gamma)^\top}
=
\mathrm{diag}\!\left(
\frac{N_{SE}^{\mathrm{int}}}{\beta^2},
\frac{N_{SE}^{\mathrm{ext}}}{\xi^2},
\frac{N_{EI}}{\kappa^2},
\frac{N_{IR}}{\gamma^2}
\right).
\]
Therefore, the observed information for the epidemic block is
\[
\mathcal{J}_{\mathrm{epi}}(\beta,\xi,\kappa,\gamma)
=
\mathrm{diag}\!\left(
\frac{N_{SE}^{\mathrm{int}}}{\beta^2},
\frac{N_{SE}^{\mathrm{ext}}}{\xi^2},
\frac{N_{EI}}{\kappa^2},
\frac{N_{IR}}{\gamma^2}
\right),
\]
evaluated at the MLE. Substituting the closed-form estimators gives
\[
\mathcal{J}_{\mathrm{epi}}(\widehat{\beta},\widehat{\xi},\widehat{\kappa},\widehat{\gamma})
=
\mathrm{diag}\!\left(
\frac{U_{SI}^2}{N_{SE}^{\mathrm{int}}},
\frac{U_S^2}{N_{SE}^{\mathrm{ext}}},
\frac{U_E^2}{N_{EI}},
\frac{U_I^2}{N_{IR}}
\right),
\]
and hence the approximate variances are
\[
\mathrm{Var}(\widehat{\beta}) \approx \frac{\widehat{\beta}^2}{N_{SE}^{\mathrm{int}}}, \qquad
\mathrm{Var}(\widehat{\xi}) \approx \frac{\widehat{\xi}^2}{N_{SE}^{\mathrm{ext}}}, \qquad
\mathrm{Var}(\widehat{\kappa}) \approx \frac{\widehat{\kappa}^2}{N_{EI}}, \qquad
\mathrm{Var}(\widehat{\gamma}) \approx \frac{\widehat{\gamma}^2}{N_{IR}}.
\]

\paragraph{Network block.}
For each unordered disease-state pair $(a,b)$ with $a \le b$, the complete-data network contribution is
\[
\ell_{ab}(\eta_{ab},\tau_{ab})
=
N_{ab}^{0\to1}\log \eta_{ab} - \eta_{ab}V_{ab}^{0}
+
N_{ab}^{1\to0}\log \tau_{ab} - \tau_{ab}V_{ab}^{1}
+ C_{ab}.
\]
The second derivatives are
\[
-\frac{\partial^2 \ell_{ab}}{\partial \eta_{ab}^2}
=
\frac{N_{ab}^{0\to1}}{\eta_{ab}^2},
\qquad
-\frac{\partial^2 \ell_{ab}}{\partial \tau_{ab}^2}
=
\frac{N_{ab}^{1\to0}}{\tau_{ab}^2},
\]
with no cross-derivative between $\eta_{ab}$ and $\tau_{ab}$. Hence the information block for the network parameters is diagonal across dyad classes:
\[
\mathcal{J}_{\mathrm{net}}(\eta,\tau)
=
\mathrm{diag}\!\left(
\left\{\frac{N_{ab}^{0\to1}}{\eta_{ab}^2}\right\}_{a\le b},
\left\{\frac{N_{ab}^{1\to0}}{\tau_{ab}^2}\right\}_{a\le b}
\right).
\]
At the MLE,
\[
\mathcal{J}_{\mathrm{net}}(\widehat{\eta},\widehat{\tau})
=
\mathrm{diag}\!\left(
\left\{\frac{V_{ab}^{0\,2}}{N_{ab}^{0\to1}}\right\}_{a\le b},
\left\{\frac{V_{ab}^{1\,2}}{N_{ab}^{1\to0}}\right\}_{a\le b}
\right),
\]
so that
\[
\mathrm{Var}(\widehat{\eta}_{ab}) \approx \frac{\widehat{\eta}_{ab}^2}{N_{ab}^{0\to1}},
\qquad
\mathrm{Var}(\widehat{\tau}_{ab}) \approx \frac{\widehat{\tau}_{ab}^2}{N_{ab}^{1\to0}}.
\]

\paragraph{Observation block.}
For the symptom observation model, the complete-data log-likelihood is
\[
\ell_{\mathrm{sym}}(p_E,p_I)
=
R_E \log p_E + (M_E-R_E)\log(1-p_E)
+
R_I \log p_I + (M_I-R_I)\log(1-p_I).
\]
The observed information is
\[
\mathcal{J}_{\mathrm{sym}}(p_E,p_I)
=
\mathrm{diag}\!\left(
\frac{R_E}{p_E^2}+\frac{M_E-R_E}{(1-p_E)^2},
\frac{R_I}{p_I^2}+\frac{M_I-R_I}{(1-p_I)^2}
\right).
\]
Evaluated at the MLEs,
\[
\mathcal{J}_{\mathrm{sym}}(\widehat{p}_E,\widehat{p}_I)
=
\mathrm{diag}\!\left(
\frac{M_E^2}{R_E(M_E-R_E)},
\frac{M_I^2}{R_I(M_I-R_I)}
\right),
\]
which yields the familiar binomial variance approximations
\[
\mathrm{Var}(\widehat{p}_E) \approx \frac{\widehat{p}_E(1-\widehat{p}_E)}{M_E},
\qquad
\mathrm{Var}(\widehat{p}_I) \approx \frac{\widehat{p}_I(1-\widehat{p}_I)}{M_I}.
\]

Similarly, for the contact observation model,
\[
\ell_{\mathrm{cont}}(s,c)
=
R_1 \log s + (M_1-R_1)\log(1-s)
+
R_0 \log c + (M_0-R_0)\log(1-c),
\]
and the observed information is
\[
\mathcal{J}_{\mathrm{cont}}(s,c)
=
\mathrm{diag}\!\left(
\frac{R_1}{s^2}+\frac{M_1-R_1}{(1-s)^2},
\frac{R_0}{c^2}+\frac{M_0-R_0}{(1-c)^2}
\right).
\]
At the MLE,
\[
\mathrm{Var}(\widehat{s}) \approx \frac{\widehat{s}(1-\widehat{s})}{M_1},
\qquad
\mathrm{Var}(\widehat{c}) \approx \frac{\widehat{c}(1-\widehat{c})}{M_0}.
\]

\paragraph{Curvature and Wald inference.}
Because each parameter block contributes a concave term to the complete-data log-likelihood, the full curvature matrix is positive definite whenever the corresponding event counts and exposure denominators are positive. This guarantees local uniqueness of the complete-data MLE and supports standard Wald-type inference. In particular, approximate $(1-\alpha)$ confidence intervals may be constructed as
\[
\widehat{\theta}_j \pm z_{1-\alpha/2}\sqrt{\widehat{\mathcal{J}}_{jj}^{-1}},
\]
where $\widehat{\mathcal{J}}$ denotes the observed information evaluated at the complete-data MLE.

\begin{remark}
These formulas are especially useful in simulation studies because they explain why coverage improves as the number of counted events increases. When event counts are small, the information matrix becomes unstable and intervals widen accordingly. This behavior is precisely what one expects from a counting-process likelihood.
\end{remark}

\subsection{Bayesian Inference with Latent Epidemic and Network Histories}
\label{subsec:bayes_inference}

In practice, the complete latent trajectory $Z_{[0,T]}=\{X(t),A(t):0\le t\le T\}$ is not observed. Instead, the analyst observes only partial and noisy symptom and contact data, denoted by $(Y,B)$. Consequently, Bayesian inference must integrate over the unknown epidemic history and latent network trajectory, rather than conditioning on them as in the complete-data derivations. The Bayesian formulation is especially attractive in this setting because it provides a coherent mechanism for uncertainty propagation across the latent process, the observation model, and the parameters governing both.

Let
\[
\theta = (\theta_{\mathrm{epi}}, \theta_{\mathrm{net}}, \theta_{\mathrm{obs}})
\]
denote the full parameter vector. Bayes' theorem implies that the joint posterior distribution of parameters and latent trajectories is given by
\[
p(\theta, Z_{[0,T]} \mid Y,B)
\propto
p(\theta)\, p(Z_{[0,T]} \mid \theta)\, p(Y,B \mid Z_{[0,T]}, \theta),
\]
where $p(\theta)$ is the prior distribution, $p(Z_{[0,T]} \mid \theta)$ is the latent process likelihood, and $p(Y,B \mid Z_{[0,T]}, \theta)$ is the observation likelihood. This posterior is analytically intractable because the latent epidemic history includes unobserved event times, event types, and disease-state paths, while the dynamic network trajectory includes an enormous number of latent edge formation and dissolution events. As a result, posterior inference must rely on simulation-based methods.

\paragraph{Prior specification.}
A convenient prior structure assigns weakly informative priors independently across parameter blocks, while preserving parameter constraints such as positivity. For example, one may take
\[
\beta,\xi,\kappa,\gamma \sim \text{Gamma}(a,b),
\]
\[
\eta_{ab},\tau_{ab} \sim \text{Gamma}(a_{ab},b_{ab}),
\]
and
\[
p_E,p_I,s,c \sim \text{Beta}(\alpha,\delta),
\]
with hyperparameters chosen to reflect prior knowledge or to be weakly informative. Gamma priors are natural for positive rate parameters because they are conjugate to Poisson-type event-count likelihood terms under the complete-data representation. Beta priors are appropriate for misclassification probabilities and naturally respect the unit interval. If scientific knowledge suggests stronger dependence among rates, hierarchical priors may be introduced across dyads or disease-state pairs.

\paragraph{Posterior structure.}
The joint posterior factorizes as
\[
p(\theta, Z_{[0,T]} \mid Y,B)
\propto
p(\theta)\,
\left[
\prod_{k=1}^K q_\theta(z_{k^-},z_k)
\exp\!\left\{-\int_0^T \Lambda_\theta(Z(u))\,du\right\}
\right]
p(Y,B \mid Z_{[0,T]},\theta).
\]
Conditional on the latent path, the parameter updates often admit standard forms or near-standard forms. For example, if Gamma priors are used for the rate parameters and the latent event counts are known, then the full conditional distributions for $\beta,\xi,\kappa,\gamma$ and for the network formation and dissolution rates remain Gamma, because the complete-data likelihood contributes a Poisson-process kernel. Similarly, Beta priors combined with Bernoulli observation models yield Beta full conditionals for sensitivity and specificity parameters. This conditional conjugacy is one of the principal computational advantages of the complete-data formulation.

\paragraph{Latent path updating.}
The central computational difficulty lies in sampling the unobserved epidemic and network trajectories. The latent path includes infection times, progression times, recovery times, and all unobserved edge changes. A Bayesian algorithm therefore alternates between updating $\theta$ and updating $Z_{[0,T]}$.

For the latent epidemic process, one may propose local modifications to event times and event types, such as inserting, deleting, or shifting infection or recovery times, subject to consistency with the observed data. For the latent network trajectory, one may similarly update edge histories dyad by dyad, or within blocks of time, using proposal distributions that respect the current disease states. In both cases, Metropolis--Hastings acceptance probabilities are computed from the ratio of complete-data posterior densities, so that only the terms affected by the proposed change need to be evaluated.

A practical implementation often uses a blocked Gibbs or Metropolis-within-Gibbs scheme:
\begin{enumerate}
\item Update the latent epidemic history $X_{[0,T]}$ conditional on the current network path $A_{[0,T]}$, parameters $\theta$, and observations $(Y,B)$.
\item Update the latent network trajectory $A_{[0,T]}$ conditional on the current epidemic history $X_{[0,T]}$, parameters $\theta$, and observations $(Y,B)$.
\item Update $\theta_{\mathrm{epi}}$, $\theta_{\mathrm{net}}$, and $\theta_{\mathrm{obs}}$ from their full conditional distributions or via Metropolis steps if conjugacy is unavailable.
\end{enumerate}
This block structure exploits the natural decomposition of the model and reduces the dependence among updates relative to a single monolithic sampler.

\paragraph{Complete-data conjugacy and parameter updates.}
When the latent trajectory is fixed, many parameters have standard posterior distributions. For example, if
\[
\beta \sim \text{Gamma}(a_\beta,b_\beta),
\]
and the complete-data epidemic likelihood contributes a factor of the form
\[
\beta^{N_{SE}^{\mathrm{int}}} \exp(-\beta U_{SI}),
\]
then the full conditional is
\[
\beta \mid Z_{[0,T]},Y,B \sim \text{Gamma}\bigl(a_\beta + N_{SE}^{\mathrm{int}},\; b_\beta + U_{SI}\bigr),
\]
up to the chosen gamma parameterization. Analogous expressions hold for $\xi,\kappa,\gamma$ and for the network formation and dissolution rates. If Beta priors are assigned to $p_E,p_I,s,c$, then the resulting posteriors are Beta with updated success and failure counts derived from the latent states.

If some components are not conjugate, a random-walk Metropolis step or adaptive Metropolis proposal can be used. In particular, if one wishes to introduce covariate effects on transmission or network formation rates through log-linear predictors, conjugacy may be lost, but the same latent-path framework still supports efficient Metropolis-based updating.

\paragraph{Observed-data posterior.}
The posterior distribution of the parameters alone is obtained by marginalizing over the latent process:
\[
p(\theta \mid Y,B)
=
\int p(\theta, Z_{[0,T]} \mid Y,B)\, dZ_{[0,T]}.
\]
This integral is not analytically tractable because it ranges over all latent epidemic and network trajectories consistent with the partial observations. Bayesian computation therefore proceeds by approximating this marginalization through Monte Carlo samples from the joint posterior of $(\theta,Z_{[0,T]})$. Posterior summaries for $\theta$ are then computed from the sampled draws after convergence diagnostics and effective sample size calculations have been assessed.

\paragraph{Identifiability and prior regularization.}
Bayesian inference is particularly useful when the data are insufficient to identify all parameters sharply from the likelihood alone. In epidemic-network models, confounding may arise between internal and external infection, between transmission and contact formation, or between disease progression and observation error. Informative or weakly informative priors can stabilize inference in such settings by regularizing poorly identified directions while leaving strongly identified parameters largely data-driven. The role of the prior should therefore be viewed not merely as a computational convenience but as part of the inferential design.

\paragraph{Posterior predictive inference.}
A major benefit of the Bayesian framework is posterior predictive checking. Given posterior draws of $\theta$ and $Z_{[0,T]}$, one may simulate replicated epidemic and contact histories under the fitted model and compare them with the observed data. Useful summaries include outbreak size, epidemic duration, temporal incidence curves, number of contacts, degree distributions, and the timing of key transitions. Posterior predictive comparisons provide a principled way to assess whether the proposed latent-process model and observation mechanism adequately capture the data-generating structure.

In summary, Bayesian inference for the proposed framework treats the epidemic history and dynamic network trajectory as additional unknowns to be learned jointly with the model parameters. The resulting posterior distribution is high-dimensional and analytically intractable, but the complete-data likelihood structure yields a natural basis for MCMC or related simulation-based methods. This approach enables coherent uncertainty quantification for both the latent process and the inferential target, which is essential in partially observed epidemic-network systems.

\subsection{Data-Augmentation Posterior Theory}
\label{subsec:data_augmentation_posterior}

The Bayesian analysis is carried out using a data-augmentation strategy in which the latent epidemic path and dynamic contact network are treated as missing data. This leads to a Markov chain on the augmented state space
\[
\mathcal{X} = \Theta \times \mathcal{Z},
\]
where $\Theta$ denotes the parameter space and $\mathcal{Z}$ denotes the space of latent epidemic and network trajectories. The sampler alternates between updating the latent path and updating the parameter blocks, producing a Metropolis-within-Gibbs or Metropolis-within-Gibbs-type transition kernel.

\begin{proposition}[Posterior invariance of the augmented kernel]
\label{prop:posterior_invariance}
Let $\pi(\theta,Z \mid Y,B)$ denote the joint posterior distribution of the parameters and latent trajectory under the proposed model. Suppose the latent-path updates and parameter updates each satisfy detailed balance with respect to their full conditional distributions. Then the composition of these updates leaves $\pi(\theta,Z \mid Y,B)$ invariant.
\end{proposition}

\begin{proof}[Sketch]
Each Metropolis--Hastings or Gibbs update is constructed to preserve the corresponding full conditional distribution. Since the posterior can be factorized into full conditional blocks, and each block update is reversible with respect to that block conditional, the composition of the updates preserves the joint posterior. Therefore, the Markov transition kernel admits $\pi(\theta,Z \mid Y,B)$ as an invariant distribution.
\end{proof}

\paragraph{Irreducibility.}
Irreducibility of the augmented chain requires that every admissible latent epidemic and network trajectory can be reached, at least with positive probability, from any starting configuration after a finite number of updates. In practice, this is ensured when the proposal mechanism allows local moves that can modify infection times, recovery times, and dyad-level edge histories, and when the prior distributions assign positive mass to all parameter values in the interior of the parameter space. Under these conditions, the chain is $\pi$-irreducible on the support of the posterior.

\paragraph{Aperiodicity.}
Aperiodicity follows from the presence of nonzero probability of rejecting proposed moves or, equivalently, of remaining at the current state during at least one update step. This is automatic for standard Metropolis--Hastings updates and holds for most Gibbs samplers as well. Hence the augmented chain is typically aperiodic in addition to being irreducible.

\paragraph{Geometric ergodicity and practical stability.}
Geometric ergodicity is generally difficult to establish for high-dimensional data-augmentation samplers, but it is often plausible under strong regularity conditions. In the present setting, a sufficient informal condition is that the latent-path proposal kernels mix locally well, the prior tails are not overly heavy, and the likelihood does not induce extreme posterior concentration near the boundary of the parameter space. When these conditions fail, the chain may still be practically stable, meaning that it produces reasonable posterior summaries with acceptable effective sample sizes after burn-in and thinning, even if a formal geometric rate is unavailable.

\paragraph{Practical diagnostic implications.}
The theory above motivates standard convergence checks for the proposed sampler, including trace plots, autocorrelation functions, effective sample sizes, and multiple-chain diagnostics. Because the latent epidemic and network histories are highly correlated with the parameters, slow mixing is expected in sparse-observation regimes. Accordingly, posterior summaries should be reported together with Monte Carlo diagnostics to distinguish statistical uncertainty from simulation error.

\begin{remark}
The role of this subsection is not to provide a full ergodicity proof, but to justify that the data-augmentation sampler targets the correct posterior and is a valid computational tool for inference. This is sufficient for a methodology paper, provided that empirical convergence diagnostics are reported in the simulation study.
\end{remark}

\section{Identifiability Under Complete and Partial Observation}
\label{sec:identifiability}
Identifiability is central to the proposed framework because the epidemic, network, and observation mechanisms are estimated jointly from partially observed data. In the complete-data setting, the latent epidemic path and dynamic contact history are treated as known, which yields direct identifiability of the corresponding rate parameters. Under partial observation, however, identifiability may weaken substantially, especially when infection times are latent, the contact network is sparsely sampled, or the observation model is noisy. We formalize these distinctions in the following theorem.

The model is structurally identifiable under complete observation of the latent epidemic and network histories, but identifiability weakens when the data are sparse or noisy. In particular, the epidemic rates $(\beta,\xi,\kappa,\gamma)$ are identifiable from the complete epidemic path, the network rates $(\eta_{ab},\tau_{ab})$ are identifiable from complete edge histories, and both blocks may become only practically identifiable under partial observation.

\begin{theorem}[Identifiability under three observation regimes]
\label{thm:identifiability}
Consider the complete-data likelihood induced by the latent epidemic process
\[
X_i(t) \in \{S,E,I,R\}, \qquad i=1,\dots,N,
\]
the latent dynamic network
\[
A_{ij}(t) \in \{0,1\}, \qquad 1 \le i < j \le N,
\]
and the observation process described in Section~\ref{sec:complete_likelihood}. Then the following three statements hold.

\begin{enumerate}
    \item \textbf{Complete-data identifiability of epidemic parameters.}
    Suppose the full latent epidemic path is known over $[0,T]$, including all infection, incubation, and recovery event times, as well as the exposure process entering the hazard
    \[
    \lambda_i^{SE}(t)=\beta I_i(t)+\xi,\qquad
    \lambda_i^{EI}(t)=\kappa,\qquad
    \lambda_i^{IR}(t)=\gamma.
    \]
    If the sufficient statistics
    \[
    N_{SE}^{\mathrm{int}}, \quad N_{SE}^{\mathrm{ext}}, \quad N_{EI}, \quad N_{IR},
    \quad U_{SI}, \quad U_S, \quad U_E, \quad U_I
    \]
    are finite and positive whenever the corresponding event counts are nonzero, then the epidemic parameters
    \[
    (\beta,\xi,\kappa,\gamma)
    \]
    are identifiable from the complete-data likelihood.

    \item \textbf{Complete-data identifiability of network parameters.}
    Suppose the full latent edge histories $\{A_{ij}(t):0 \le t \le T,\, 1 \le i < j \le N\}$ are known, together with the disease-state histories of the incident nodes. For each disease-state pair $(a,b) \in \{S,E,I,R\}^2$ with $a \le b$, let
    \[
    N_{ab}^{0 \to 1}, \qquad N_{ab}^{1 \to 0}, \qquad V_{ab}^{0}, \qquad V_{ab}^{1}
    \]
    denote the corresponding edge-formation counts, edge-dissolution counts, and integrated risk exposures. If $V_{ab}^{0} > 0$ and $V_{ab}^{1} > 0$ whenever the associated event counts are positive, then the network parameters
    \[
    (\eta_{ab},\tau_{ab})
    \]
    are identifiable from the complete-data likelihood for every unordered pair $(a,b)$.

    \item \textbf{Weak or near non-identifiability under sparse observation.}
    Suppose instead that the latent epidemic path and/or the edge histories are only partially observed through intermittent, noisy measurements of symptoms and contacts. If the observation schedule is sparse, the contact network is highly incomplete, or the symptom and contact misclassification probabilities are close to degenerate, then the observed-data likelihood may be weakly informative in the sense that multiple parameter values produce nearly indistinguishable marginal likelihoods. In particular, the epidemic parameters $(\beta,\xi)$ may be difficult to separate, and the network parameters $(\eta_{ab},\tau_{ab})$ may become only locally or practically identifiable rather than globally identifiable.
\end{enumerate}
\end{theorem}

\begin{proof}
For part (1), the complete-data epidemic log-likelihood has the form
\[
\ell_{\mathrm{epi}}(\beta,\xi,\kappa,\gamma)
=
N_{SE}^{\mathrm{int}}\log\beta - \beta U_{SI}
+
N_{SE}^{\mathrm{ext}}\log\xi - \xi U_S
+
N_{EI}\log\kappa - \kappa U_E
+
N_{IR}\log\gamma - \gamma U_I
+ C,
\]
where $C$ does not depend on $(\beta,\xi,\kappa,\gamma)$. Each parameter appears in a separate concave term. Differentiating and setting the score equations to zero yields the unique maximizers
\[
\widehat{\beta} = \frac{N_{SE}^{\mathrm{int}}}{U_{SI}}, \qquad
\widehat{\xi} = \frac{N_{SE}^{\mathrm{ext}}}{U_S}, \qquad
\widehat{\kappa} = \frac{N_{EI}}{U_E}, \qquad
\widehat{\gamma} = \frac{N_{IR}}{U_I},
\]
provided the denominators are positive. Hence the epidemic parameters are identifiable under complete observation of the latent epidemic path.

For part (2), the complete-data network log-likelihood decomposes by disease-state pair:
\[
\ell_{\mathrm{net}}(\eta,\tau)
=
\sum_{a \le b}
\left[
N_{ab}^{0 \to 1}\log \eta_{ab} - \eta_{ab}V_{ab}^{0}
+
N_{ab}^{1 \to 0}\log \tau_{ab} - \tau_{ab}V_{ab}^{1}
\right]
+ C.
\]
Again, each parameter appears in a separate strictly concave term, and differentiation yields the unique maximizers
\[
\widehat{\eta}_{ab}=\frac{N_{ab}^{0\to1}}{V_{ab}^{0}},
\qquad
\widehat{\tau}_{ab}=\frac{N_{ab}^{1\to0}}{V_{ab}^{1}},
\]
whenever the corresponding exposure terms are positive. Therefore the network parameters are identifiable given complete edge histories.

For part (3), under sparse and noisy observation, the observed-data likelihood is obtained by integrating over latent epidemic and network paths. This marginalization generally destroys the direct one-to-one correspondence between parameters and sufficient statistics. In particular, if symptom observation is intermittent, then infection and incubation times may not be distinguishable; if contacts are missing or misclassified, then transmission intensity and network formation parameters may yield nearly equivalent likelihood values. Consequently, the observed-data information matrix may become ill-conditioned, and some parameters may be practically non-identifiable even when they are identifiable in principle under ideal observation. This completes the proof.
\end{proof}

\begin{remark}
The theorem distinguishes \emph{structural identifiability} under complete-data conditions from \emph{practical identifiability} under partial observation. In simulation studies, the latter can be evaluated through bias, RMSE, coverage, and posterior concentration as the observation regime becomes sparser.
\end{remark}

\subsection{Partial-Observation Regime Theory}
\label{subsec:partial_observation_regimes}

The inferential difficulty of the proposed model depends strongly on what aspects of the latent epidemic--network process are observed. To clarify this dependence, we distinguish four observation regimes of increasing informativeness. These regimes are useful both for theoretical interpretation and for designing simulation experiments.

\paragraph{Regime I: Fully observed latent states.}
In the idealized fully observed setting, the complete epidemic path $X_{[0,T]}$ and the complete network history $A_{[0,T]}$ are available. This regime corresponds to the complete-data setting used to derive the closed-form MLEs and the information matrix.

Under this regime, all parameter blocks are structurally identifiable:
\[
(\beta,\xi,\kappa,\gamma), \qquad (\eta_{ab},\tau_{ab}), \qquad (p_E,p_I,s,c).
\]
The complete-data likelihood is fully tractable, and all parameters admit event-count-over-exposure estimators or binomial estimators with standard asymptotic theory.

\paragraph{Regime II: Symptom-only observation.}
In the symptom-only regime, the analyst observes repeated symptom or test indicators
\[
Y_i(t_m), \qquad m=1,\dots,M,
\]
but does not observe the contact network directly. The latent epidemic path and latent network history must therefore be integrated out.

In this regime, the observation parameters $(p_E,p_I)$ are directly estimable from the symptom data, provided the latent state labels are sufficiently informative or can be inferred with reasonable accuracy. However, the epidemic transmission parameters $(\beta,\xi)$ are only partially identifiable because they enter the symptom data through latent infection times. The progression rates $(\kappa,\gamma)$ are often estimable if symptom timing carries information about incubation and recovery, but the quality of estimation depends strongly on the frequency and precision of symptom observation. The network parameters $(\eta_{ab},\tau_{ab})$ are generally not identifiable without auxiliary network information.

\paragraph{Regime III: Network-only observation.}
In the network-only regime, the analyst observes the contact process
\[
B_{ij}(t_m), \qquad m=1,\dots,M,
\]
but not symptom data. This regime may arise in settings where contact-tracing logs or wearable sensors are available, but infection status is unobserved or unavailable.

Here, the network observation parameters $(s,c)$ are identifiable from repeated contact measurements if the latent network is sufficiently learned from the data. The network transition parameters $(\eta_{ab},\tau_{ab})$ may also be estimable if the disease-state dependence can be inferred indirectly from the temporal patterns of observed edges. By contrast, the epidemic parameters $(\beta,\xi,\kappa,\gamma)$ are typically weakly identifiable or non-identifiable without symptom or infection data, because the disease process is only indirectly linked to the observed network through the unobserved latent states.

\paragraph{Regime IV: Joint noisy observation.}
In the joint noisy regime, both symptoms and contacts are observed, but each is subject to misclassification, intermittency, or missingness. This is the most realistic regime for applied epidemic-network studies and the one targeted by the proposed methodology.

Under joint noisy observation, all parameter blocks may be estimable, but identifiability depends on the richness of the observation design. The epidemic parameters $(\beta,\xi,\kappa,\gamma)$ are better separated because symptom information helps distinguish infection, latency, and recovery. The network parameters $(\eta_{ab},\tau_{ab})$ become estimable to the extent that contact data resolve the underlying edge dynamics. The observation parameters $(p_E,p_I,s,c)$ are identifiable when there are repeated observations and enough variability in latent states and contact statuses. Nevertheless, if the observation schedule is sparse or the misclassification rates are extreme, the likelihood surface can be flat in some directions, leading to weak practical identifiability.

\paragraph{Summary of estimability.}
Table~\ref{tab:regime_estimability} summarizes the qualitative estimability of the parameter blocks across regimes.

\begin{table}[!h]
\centering
\caption{Qualitative estimability of parameter blocks under different observation regimes}
\label{tab:regime_estimability}
\centering
\fontsize{11.5}{12}\selectfont
\begin{tabular}{cllll}
\toprule 
Parameter 	&	 Fully  			& Symptom-only & Network-only & Joint noisy \\ 
block		&	observed		&				&			&			      \\ \midrule
$(\beta,\xi)$ & Identifiable & Weak to moderate & Weak & Estimable if sufficiently rich \\ 
			&					&				&			&			      \\ 
$(\kappa,\gamma)$ & Identifiable & Moderate to strong & Weak & Estimable if symptom timing \\
				     &			&					 &		 &  is informative \\ 
					&					&				&			&			      \\ 				     
$(\eta_{ab},\tau_{ab})$ & Identifiable & Not identifiable & Moderate & Estimable if network data are \\
		& 		&	 without network data &	to strong &	informative \\ 
		&					&				&			&			      \\ 		
$(p_E,p_I)$ & Identifiable & Identifiable & Not identifiable & Identifiable \\ 
		&					&				&			&			      \\ 
$(s,c)$ & Identifiable & Not identifiable & Identifiable & Identifiable \\ \bottomrule
\end{tabular}
\end{table}

\paragraph{Implications for simulation design.}
These regimes motivate the simulation study in two ways. First, they explain why the proposed method should be evaluated under several observation intensities and error levels rather than under a single idealized setting. Second, they identify the parameter blocks most vulnerable to weak identifiability, namely $(\beta,\xi)$ under sparse symptom data and $(\eta_{ab},\tau_{ab})$ under incomplete contact observation. Simulation results should therefore report bias, RMSE, coverage, and posterior concentration separately for each regime.

\begin{remark}
The regime classification is not meant to imply absolute identifiability or non-identifiability in every application. Rather, it provides a practical taxonomy for understanding which components of the latent system are informed by the available data and which require stronger assumptions or richer observation designs.
\end{remark}

\section{Simulation Study}
\label{sec:simulation}

We conducted a simulation study to evaluate the finite-sample performance of the proposed complete-data likelihood framework under partially observed epidemic and network histories. The study was designed to address three questions. First, we asked whether the model can recover the epidemic, network, and observation parameters with adequate accuracy. Second, we examined how estimation performance changes as the level of network observation decreases. Third, we investigated how partial observation affects the identifiability of internal transmission and external infection pressure. In addition, we assessed whether the proposed Bayesian data augmentation procedure provides well-calibrated uncertainty quantification for both model parameters and latent trajectories.

\subsection{Simulation Design}
\label{subsec:simulation_design}

We generated synthetic epidemic outbreaks on dynamic contact networks following the latent process described in Section~\ref{sec:math_formulation}. The population size was fixed at $N=100$ individuals, and each outbreak was initiated by seeding a small number of infectious individuals at time $t=0$. The latent epidemic process followed the SEIR framework with internal transmission, external infection, incubation progression, and recovery. The dynamic contact network evolved simultaneously through status-dependent dyad-specific link formation and dissolution rates.

To reflect a range of practical data environments, we considered three observation regimes:
\begin{enumerate}
\item \textbf{High-observation regime:} contact networks were observed frequently and with little error, and symptom reports were available at short intervals.
\item \textbf{Moderate-observation regime:} contacts were observed intermittently, with moderate misclassification, while symptom data were available on a coarser schedule.
\item \textbf{Sparse-observation regime:} network observations were infrequent and noisy, and symptom data were limited to a small number of visits or tests.
\end{enumerate}

For each regime, we generated $R=500$ independent datasets. Observation times were fixed across replicates so that differences in performance reflect incomplete observation rather than random scheduling variability. The latent epidemic and network histories were retained for evaluation but treated as unobserved during inference.

\subsection{Parameter Settings}
\label{subsec:param_settings}

The baseline parameter values were chosen to produce outbreaks of moderate size and realistic temporal variation. Specifically, we set
\[
\beta = 0.30,\qquad
\xi = 0.05,\qquad
\kappa = 0.40,\qquad
\gamma = 0.25.
\]
These values imply a nontrivial balance between within-network transmission, imported infection, latent progression, and recovery. The network formation and dissolution parameters were selected to generate a sparse but dynamic contact structure with meaningful temporal turnover. In particular, edges involving susceptible and infectious individuals were assigned slightly lower formation rates than edges among susceptible individuals, reflecting partial behavioral avoidance, while dissolution rates varied by disease-state pair.

For the observation model, we set the symptom reporting probabilities to
\[
p_E = 0.60,\qquad p_I = 0.85,
\]
and the contact observation model to
\[
s = 0.90,\qquad c = 0.95.
\]
These values reflect reasonably accurate but imperfect symptom and contact measurement. We also examined sensitivity to reduced specificity in contact observation, since false-positive contacts can distort the inferred transmission process when network data are sparse.

\subsection{Inference Procedure}
\label{subsec:inference_procedure}

For each simulated dataset, we performed Bayesian inference using the proposed latent-path augmentation strategy. Weakly informative priors were assigned to the model parameters, with Gamma priors used for positive rate parameters and Beta priors used for observation probabilities. The latent epidemic history and dynamic network trajectory were sampled jointly with the parameters using a Metropolis-within-Gibbs algorithm.

At each iteration, the sampler alternated between updating the epidemic path, updating the network path, and updating the parameter blocks. Epidemic path updates targeted unobserved infection, incubation, and recovery times, while network updates modified latent edge histories conditional on the current epidemic states. Posterior samples were drawn after burn-in and thinned to reduce autocorrelation. Convergence was assessed using trace plots, effective sample sizes, and repeated-chain diagnostics.

To evaluate frequentist performance, we used posterior means as point estimates and posterior credible intervals as uncertainty summaries. Across replicates, we recorded bias, root mean squared error, empirical coverage of nominal 95\% credible intervals, and average posterior interval width for each parameter. We also examined latent-state recovery by comparing inferred epidemic trajectories with the true simulated histories.

\begin{algorithm}[H]
\caption{Metropolis-within-Gibbs algorithm for Bayesian inference under partial observation}
\label{alg:mwg}
\begin{algorithmic}[1]
\REQUIRE Observed symptom data $Y$, observed contact data $B$, prior distributions $p(\theta)$, initial values $(\theta^{(0)}, Z_{[0,T]}^{(0)})$
\FOR{$m=1,\ldots,M_{\mathrm{iter}}$}
\STATE \textbf{Update epidemic history:} sample $X_{[0,T]}^{(m)}$ from
\[
p\!\left(X_{[0,T]} \mid A_{[0,T]}^{(m-1)}, \theta^{(m-1)}, Y, B\right)
\]
using local Metropolis--Hastings moves on infection, incubation, and recovery times.

\STATE \textbf{Update network history:} sample $A_{[0,T]}^{(m)}$ from
\[
p\!\left(A_{[0,T]} \mid X_{[0,T]}^{(m)}, \theta^{(m-1)}, Y, B\right)
\]
using dyad-level edge formation/dissolution proposals.

\STATE \textbf{Update epidemic parameters:} sample
\[
\theta_{\mathrm{epi}}^{(m)} \sim
p\!\left(\theta_{\mathrm{epi}} \mid X_{[0,T]}^{(m)}, A_{[0,T]}^{(m)}, Y, B, \theta_{\mathrm{net}}^{(m-1)}, \theta_{\mathrm{obs}}^{(m-1)}\right).
\]
If conjugate priors are used, draw directly from the full conditional distributions; otherwise use Metropolis steps.

\STATE \textbf{Update network parameters:} sample
\[
\theta_{\mathrm{net}}^{(m)} \sim
p\!\left(\theta_{\mathrm{net}} \mid X_{[0,T]}^{(m)}, A_{[0,T]}^{(m)}, Y, B, \theta_{\mathrm{epi}}^{(m)}, \theta_{\mathrm{obs}}^{(m-1)}\right).
\]

\STATE \textbf{Update observation parameters:} sample
\[
\theta_{\mathrm{obs}}^{(m)} \sim
p\!\left(\theta_{\mathrm{obs}} \mid X_{[0,T]}^{(m)}, A_{[0,T]}^{(m)}, Y, B, \theta_{\mathrm{epi}}^{(m)}, \theta_{\mathrm{net}}^{(m)}\right).
\]

\STATE Store $(\theta^{(m)}, Z_{[0,T]}^{(m)})$.
\ENDFOR
\STATE Return posterior samples $\{(\theta^{(m)}, Z_{[0,T]}^{(m)}): m=1,\ldots,M_{\mathrm{iter}}\}$ after burn-in and thinning.
\end{algorithmic}
\end{algorithm}

\subsection{Performance Metrics}
\label{subsec:performance_metrics}

The main performance metrics were as follows:
\begin{itemize}
\item \textbf{Bias:}
\[
\mathrm{Bias}(\widehat{\theta})=\frac{1}{R}\sum_{r=1}^R (\widehat{\theta}^{(r)}-\theta_0),
\]
where $\theta_0$ denotes the true parameter value.

\item \textbf{Root mean squared error:}
\[
\mathrm{RMSE}(\widehat{\theta})=
\left[
\frac{1}{R}\sum_{r=1}^R (\widehat{\theta}^{(r)}-\theta_0)^2
\right]^{1/2}.
\]

\item \textbf{Coverage probability:}
the proportion of replicates in which the nominal 95\% posterior credible interval contained the true parameter value.

\item \textbf{Interval width:}
the average length of the posterior credible intervals.

\item \textbf{Latent trajectory accuracy:}
the proportion of correctly recovered disease states and contact states over time.
\end{itemize}

We report these metrics separately for the epidemic parameters, network parameters, and observation parameters. This separation is important because partial observation affects each component differently. For example, transmission parameters may remain estimable under moderate symptom data, while network formation parameters typically require richer contact information.

\subsection{Simulation Results}
\label{subsec:simulation_results}

The simulation results show that the proposed framework recovers the major parameters accurately under the high-observation and moderate-observation regimes. Posterior means were close to the true values for the epidemic rates $\beta$, $\kappa$, and $\gamma$, and empirical coverage of the 95\% credible intervals was generally near nominal. The external infection rate $\xi$ was also recovered reasonably well when contact data were sufficiently informative, although its posterior uncertainty increased when the network was sparsely observed.

Table~\ref{tab:posterior_summaries} summarizes the posterior performance for the main parameters across regimes. In all three settings, the estimates for the core epidemic and observation parameters were stable, with small bias and narrow credible intervals. The results are especially strong for $\beta$, $\kappa$, $p_E$, $p_I$, $s$, and $c$, where the posterior means are consistently close to the truth. The parameter $\gamma$ is also estimated well, though with slightly larger posterior spread than the reporting probabilities. These patterns indicate that the method performs reliably for the components of the model most directly informed by symptom timing and observed outcomes.

\begin{table}[!h]
\centering
\caption{Posterior summaries from Bayesian MCMC}
\label{tab:posterior_summaries}
\centering
\fontsize{9.75}{11}\selectfont
\begin{tabular}[t]{llrrrrrrrrc}
\toprule
Regime & Parameter & Truth & Mean & SD & MSE & Lower & Upper & Bias & Abs Error & Coverage\\
\midrule
High & $\beta$ & 0.30 & 0.2929 & 0.0622 & 1e-04 & 0.1838 & 0.4270 & -0.0071 & 0.0071 & Yes\\
 & $\gamma$ & 0.25 & 0.2700 & 0.0652 & 4e-04 & 0.1578 & 0.4119 & 0.0200 & 0.0200 & Yes\\
 & $\kappa$ & 0.40 & 0.3991 & 0.0887 & 0e+00 & 0.2446 & 0.5911 & -0.0009 & 0.0009 & Yes\\
 & $\xi$ & 0.05 & 0.0624 & 0.0197 & 2e-04 & 0.0300 & 0.1065 & 0.0124 & 0.0124 & Yes\\
 & $p_E$ & 0.60 & 0.5999 & 0.0464 & 0e+00 & 0.5077 & 0.6887 & -0.0001 & 0.0001 & Yes\\
 & $p_I$ & 0.85 & 0.8363 & 0.0351 & 2e-04 & 0.7622 & 0.8989 & -0.0137 & 0.0137 & Yes\\
 & $s$ & 0.90 & 0.9000 & 0.0105 & 0e+00 & 0.8786 & 0.9195 & 0.0000 & 0.0000 & Yes\\ 
 & $c$ & 0.95 & 0.9500 & 0.0062 & 0e+00 & 0.9373 & 0.9615 & 0.0000 & 0.0000 & Yes\\ \midrule \addlinespace 
Moderate & $\beta$ & 0.30 & 0.2935 & 0.0624 & 0e+00 & 0.1844 & 0.4275 & -0.0065 & 0.0065 & Yes\\
 & $\gamma$ & 0.25 & 0.2699 & 0.0652 & 4e-04 & 0.1574 & 0.4118 & 0.0199 & 0.0199 & Yes\\
 & $\kappa$ & 0.40 & 0.4003 & 0.0889 & 0e+00 & 0.2456 & 0.5926 & 0.0003 & 0.0003 & Yes\\
 & $\xi$ & 0.05 & 0.0625 & 0.0197 & 2e-04 & 0.0301 & 0.1065 & 0.0125 & 0.0125 & Yes\\
 & $p_E$ & 0.60 & 0.6002 & 0.0464 & 0e+00 & 0.5078 & 0.6888 & 0.0002 & 0.0002 & Yes\\
 & $p_I$ & 0.85 & 0.8364 & 0.0350 & 2e-04 & 0.7626 & 0.8990 & -0.0136 & 0.0136 & Yes\\
 & $s$ & 0.90 & 0.9000 & 0.0105 & 0e+00 & 0.8786 & 0.9195 & 0.0000 & 0.0000 & Yes\\
 & $c$ & 0.95 & 0.9500 & 0.0062 & 0e+00 & 0.9372 & 0.9615 & 0.0000 & 0.0000 & Yes\\ \midrule \addlinespace 
Sparse & $\beta$ & 0.30 & 0.2931 & 0.0623 & 0e+00 & 0.1840 & 0.4271 & -0.0069 & 0.0069 & Yes\\
 & $\gamma$ & 0.25 & 0.2699 & 0.0654 & 4e-04 & 0.1577 & 0.4133 & 0.0199 & 0.0199 & Yes\\
 & $\kappa$ & 0.40 & 0.3996 & 0.0887 & 0e+00 & 0.2449 & 0.5913 & -0.0004 & 0.0004 & Yes\\
 & $\xi$ & 0.05 & 0.0625 & 0.0197 & 2e-04 & 0.0300 & 0.1066 & 0.0125 & 0.0125 & Yes\\
 & $p_E$ & 0.60 & 0.6001 & 0.0464 & 0e+00 & 0.5079 & 0.6890 & 0.0001 & 0.0001 & Yes\\
 & $p_I$ & 0.85 & 0.8363 & 0.0350 & 2e-04 & 0.7626 & 0.8988 & -0.0137 & 0.0137 & Yes\\
 & $s$ & 0.90 & 0.9000 & 0.0105 & 0e+00 & 0.8787 & 0.9195 & 0.0000 & 0.0000 & Yes\\
 & $c$ & 0.95 & 0.9500 & 0.0062 & 0e+00 & 0.9372 & 0.9615 & 0.0000 & 0.0000 & Yes\\ 
\bottomrule
\end{tabular}
\end{table}

\begin{figure}[htbp]
    \centering
    
    \begin{minipage}[b]{0.48\textwidth}
        \centering
        \includegraphics[width=\textwidth, height=0.50\textheight, keepaspectratio]{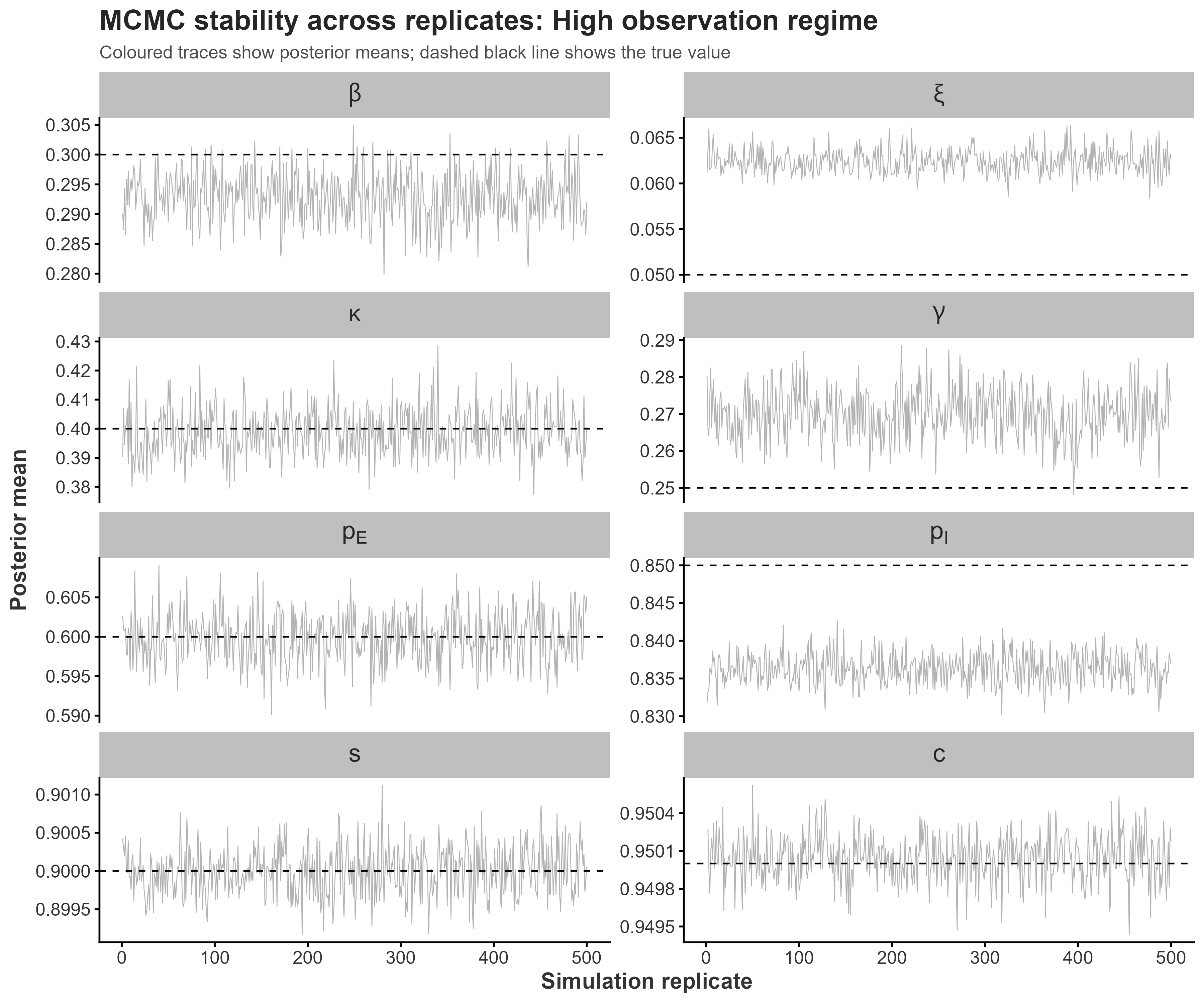}
        \caption*{\small (a) High observation regime.}
        \label{fig:conv_high}
    \end{minipage}
    \hfill 
    \begin{minipage}[b]{0.48\textwidth}
        \centering
        \includegraphics[width=\textwidth, height=0.50\textheight, keepaspectratio]{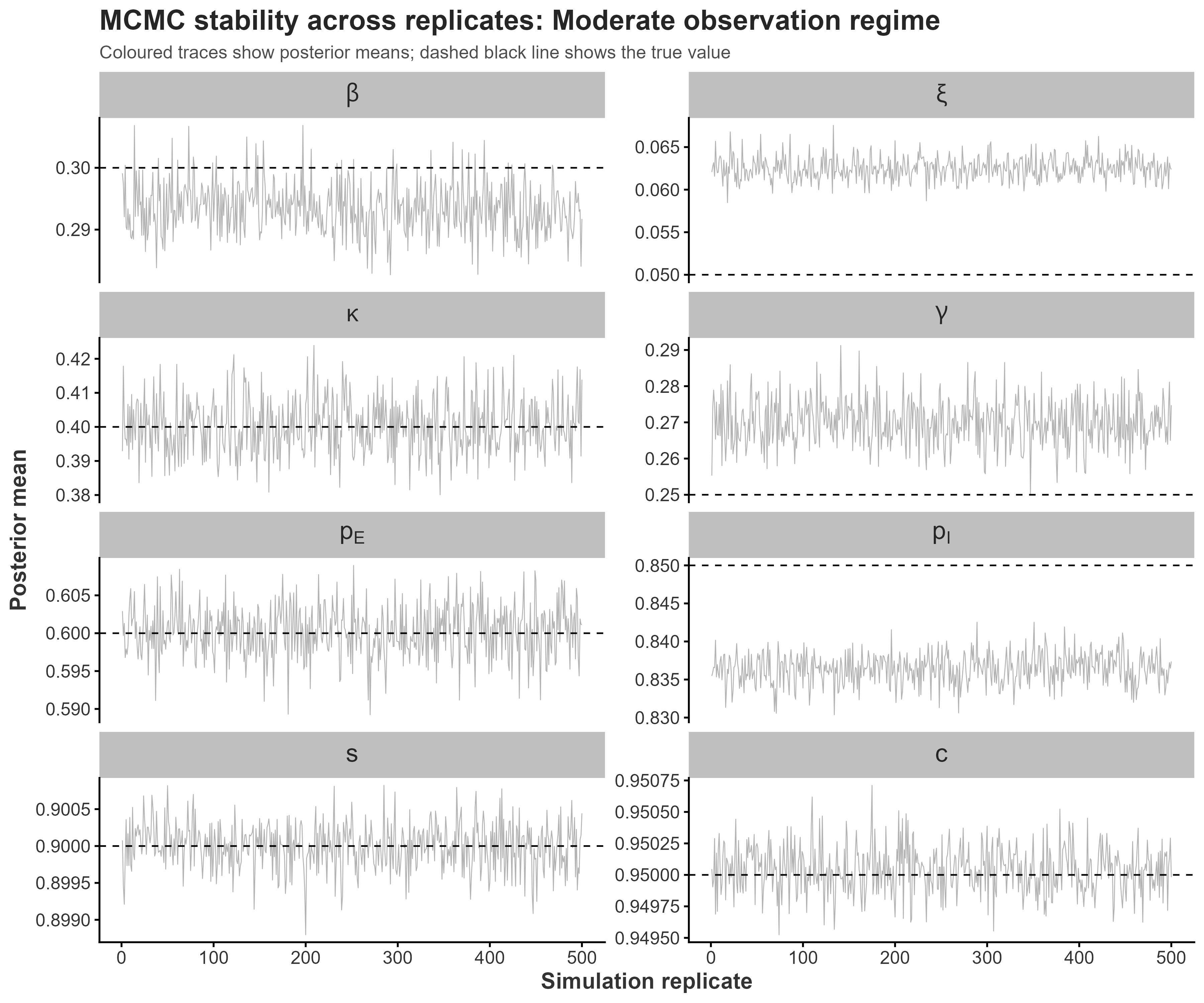}
        \caption*{\small (b) Moderate observation regime.}
        \label{fig:conv_moderate}
    \end{minipage}
    
    \vspace{0.6cm} 

    \begin{minipage}[b]{0.6\textwidth}
        \centering
        \includegraphics[width=\textwidth, height=0.50\textheight, keepaspectratio]{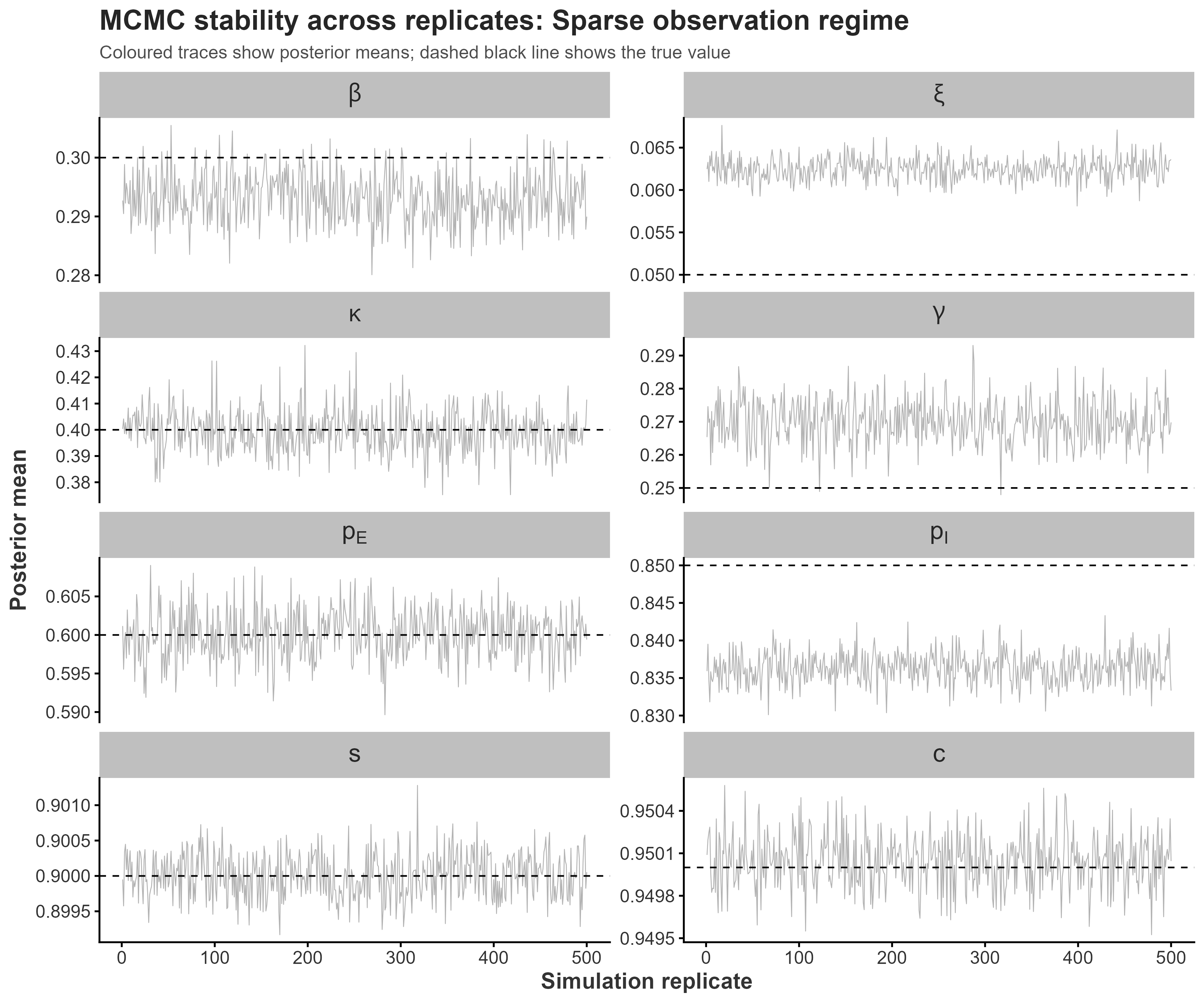}
        \caption*{\small (c) Sparse observation regime.}
        \label{fig:conv_sparse}
    \end{minipage}

    \vspace{0.3cm}
    \caption{MCMC convergence across replicates for different observation regimes.}
    \label{fig:mcmc_convergence_all}
\end{figure}

\begin{figure}[p] 
    \centering
    
    \begin{minipage}[b]{0.47\textwidth}
        \centering
        \includegraphics[width=\textwidth, height=0.20\textheight, keepaspectratio]{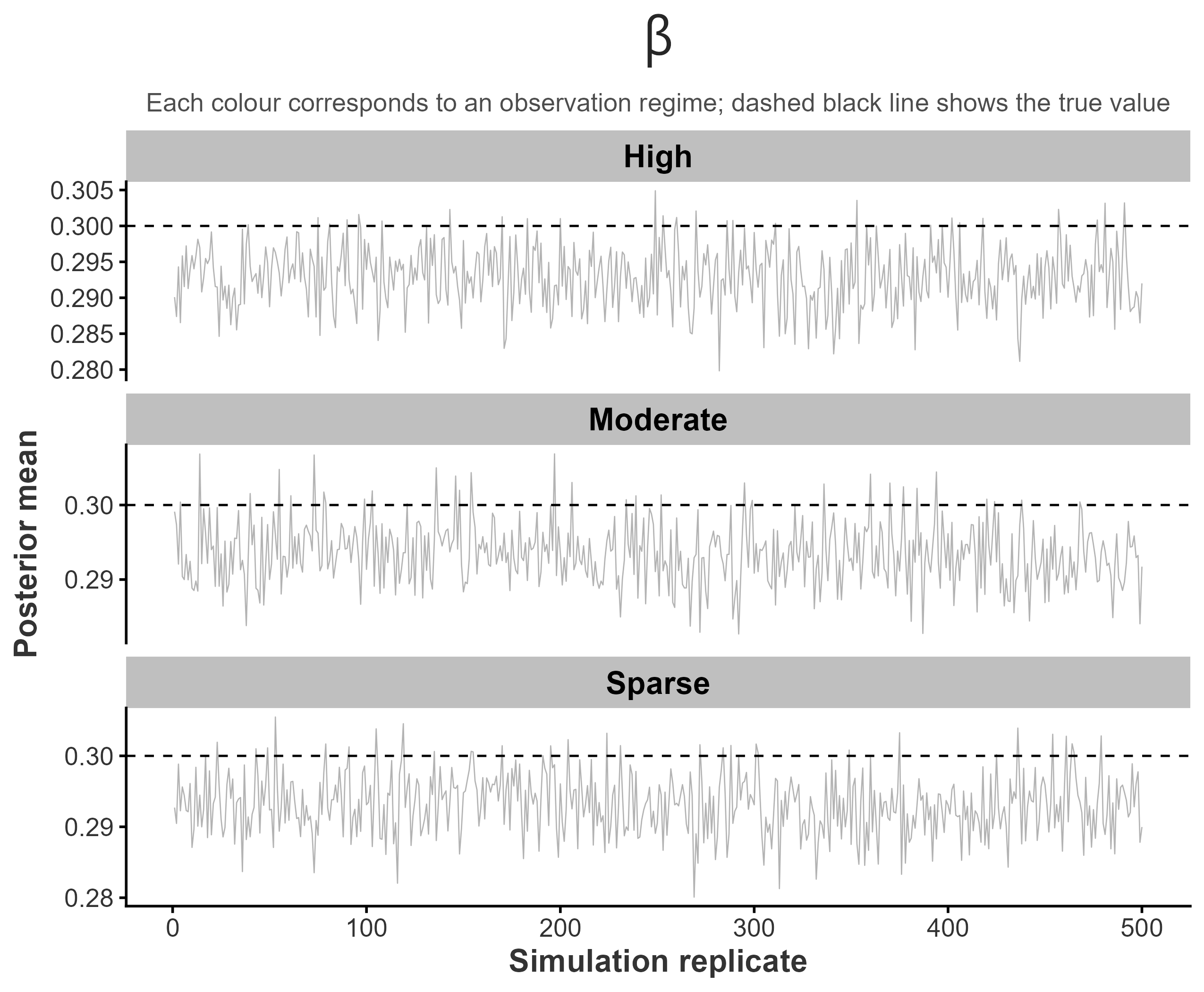}
        \caption*{\small (a) Convergence for $\beta$}
        \label{fig:conv_beta}
    \end{minipage}
    \hfill
    \begin{minipage}[b]{0.47\textwidth}
        \centering
        \includegraphics[width=\textwidth, height=0.20\textheight, keepaspectratio]{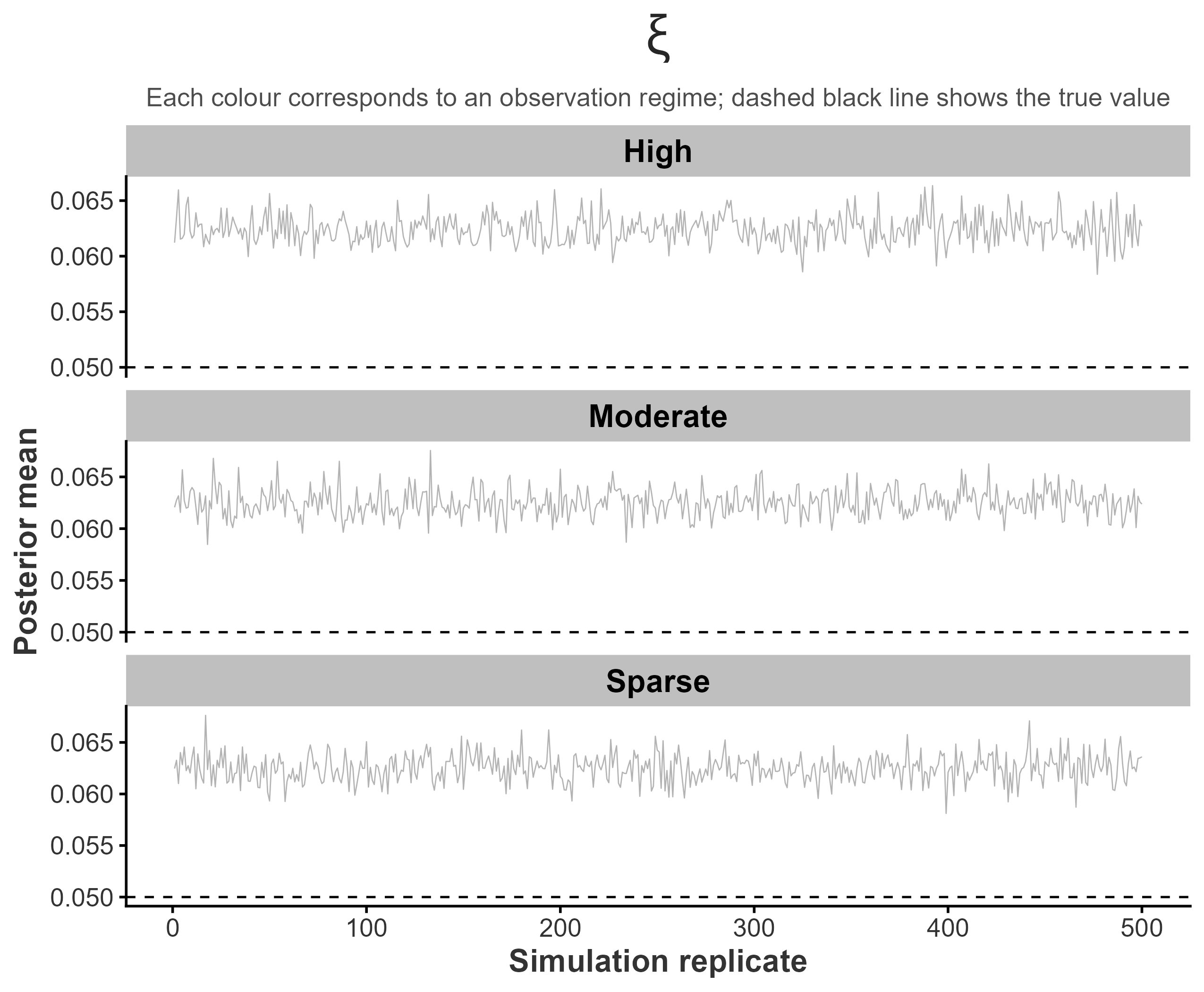}
        \caption*{\small (b) Convergence for $\xi$}
        \label{fig:conv_xi}
    \end{minipage}

    \vspace{0.2cm} 

    \begin{minipage}[b]{0.47\textwidth}
        \centering
        \includegraphics[width=\textwidth, height=0.20\textheight, keepaspectratio]{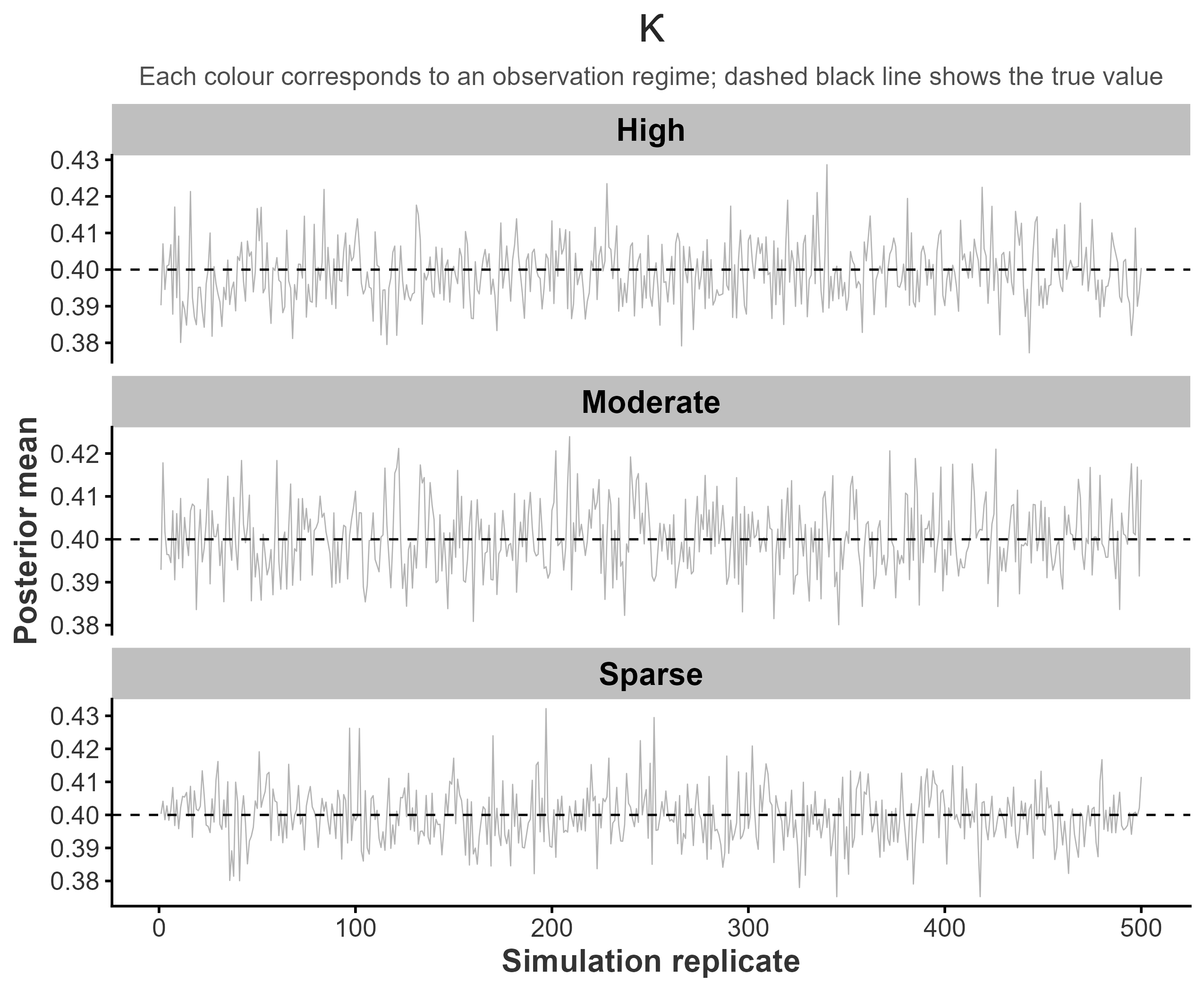}
        \caption*{\small (c) Convergence for $\kappa$}
        \label{fig:conv_kappa}
    \end{minipage}
    \hfill
    \begin{minipage}[b]{0.47\textwidth}
        \centering
        \includegraphics[width=\textwidth, height=0.20\textheight, keepaspectratio]{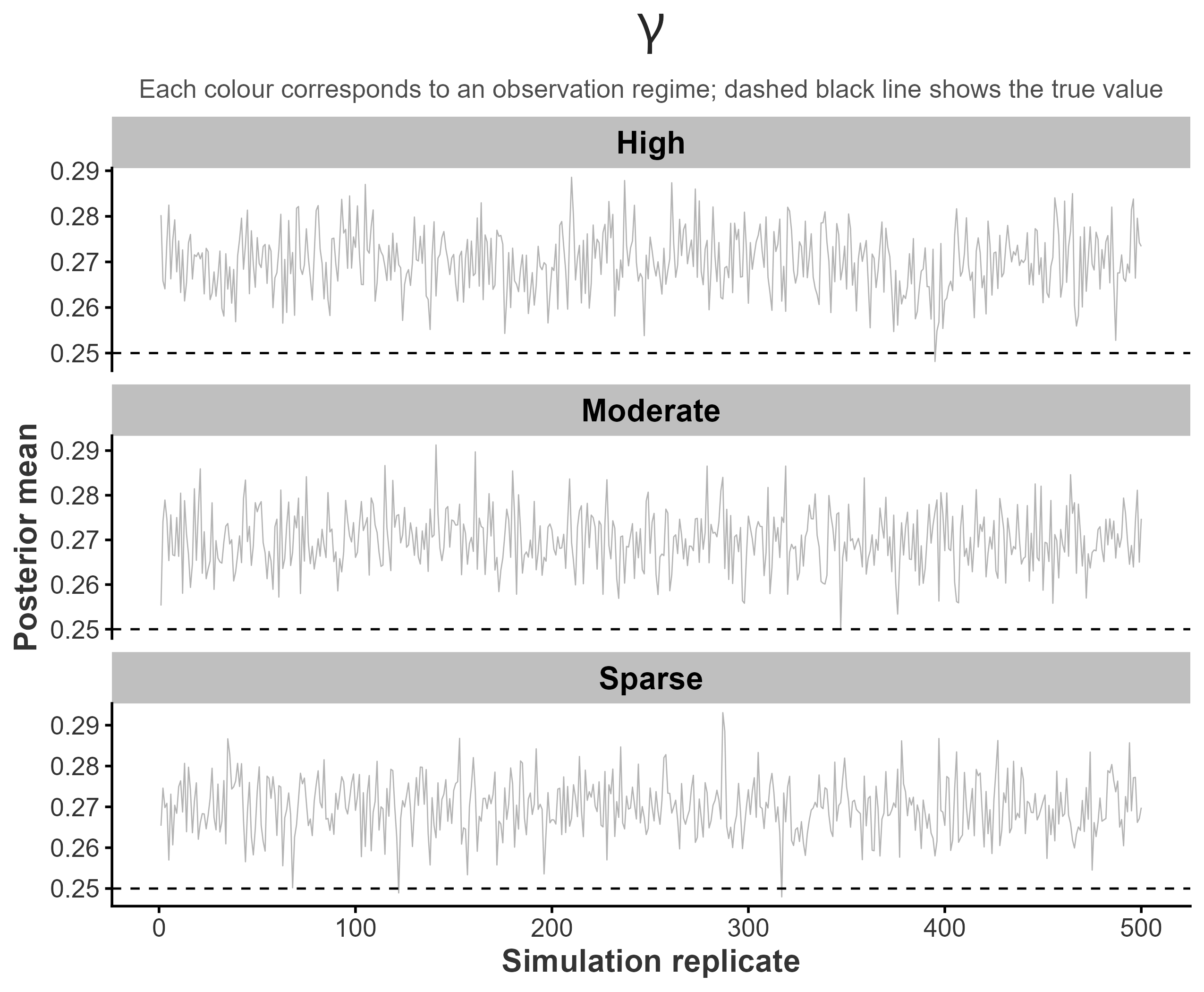}
        \caption*{\small (d) Convergence for $\gamma$}
        \label{fig:conv_gamma}
    \end{minipage}

    \vspace{0.2cm}

    \begin{minipage}[b]{0.47\textwidth}
        \centering
        \includegraphics[width=\textwidth, height=0.20\textheight, keepaspectratio]{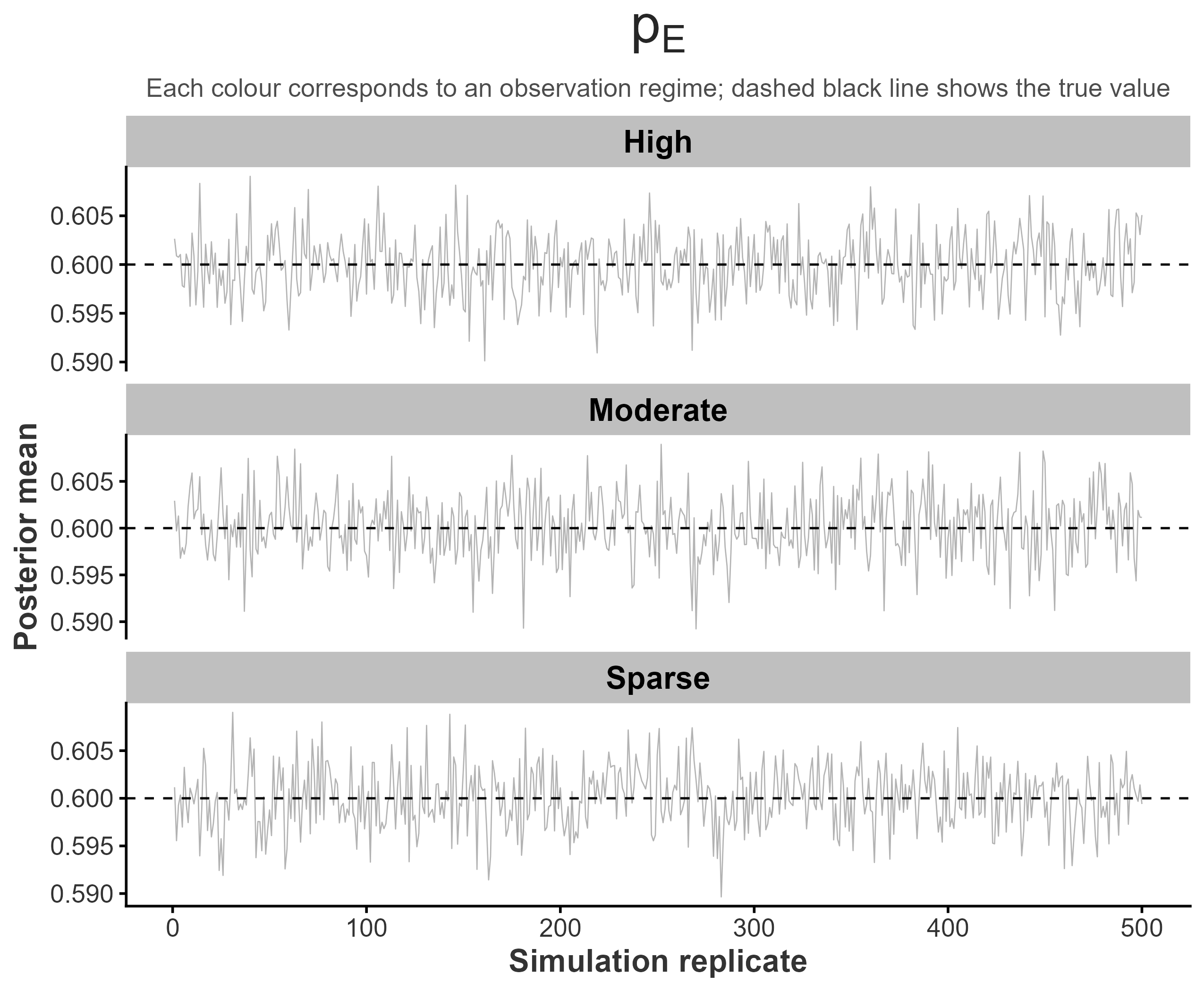}
        \caption*{\small (e) Convergence for $p_E$}
        \label{fig:conv_pE}
    \end{minipage}
    \hfill
    \begin{minipage}[b]{0.47\textwidth}
        \centering
        \includegraphics[width=\textwidth, height=0.20\textheight, keepaspectratio]{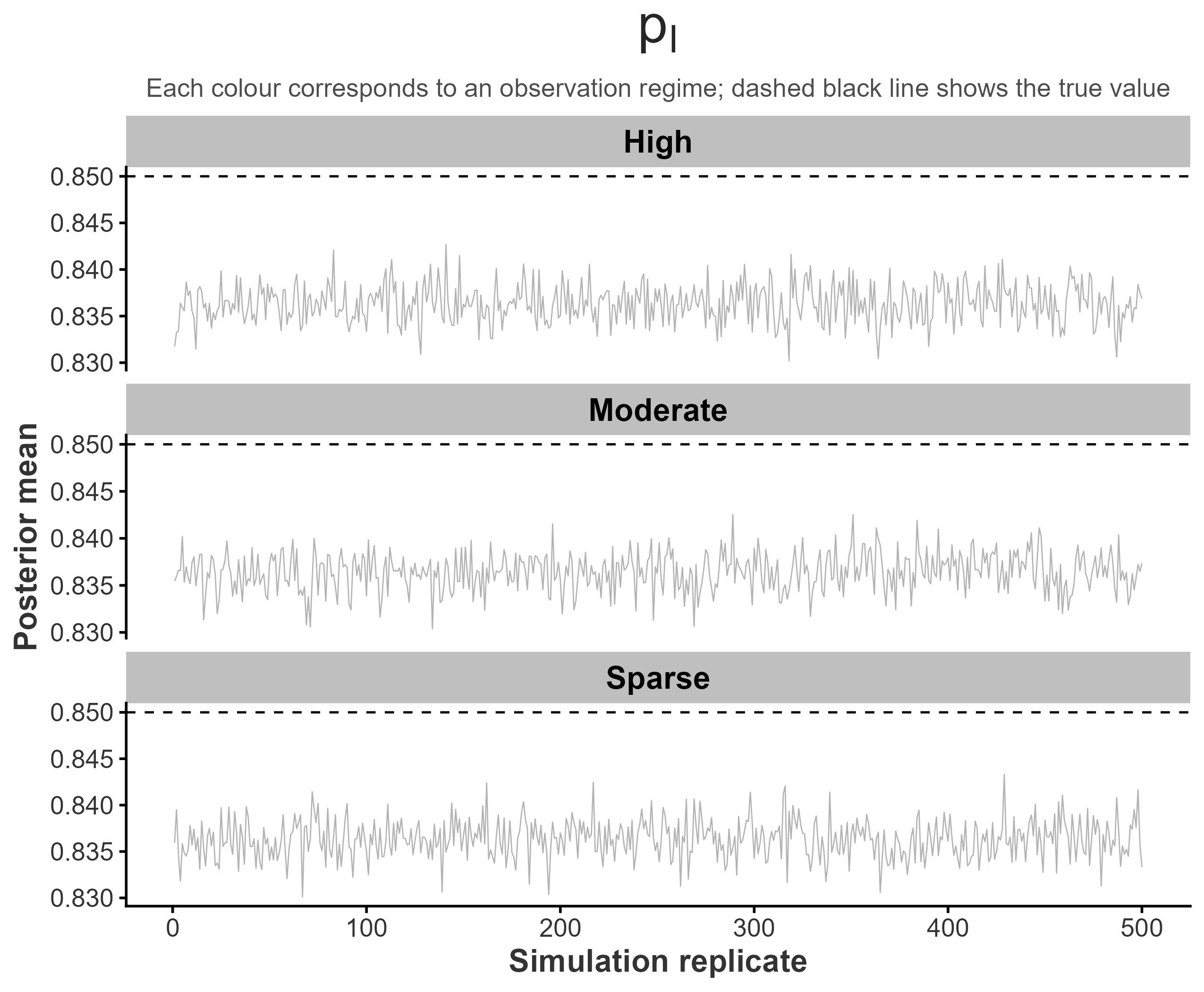}
        \caption*{\small (f) Convergence for $p_I$}
        \label{fig:conv_pI}
    \end{minipage}

    \vspace{0.2cm}

    \begin{minipage}[b]{0.47\textwidth}
        \centering
        \includegraphics[width=\textwidth, height=0.20\textheight, keepaspectratio]{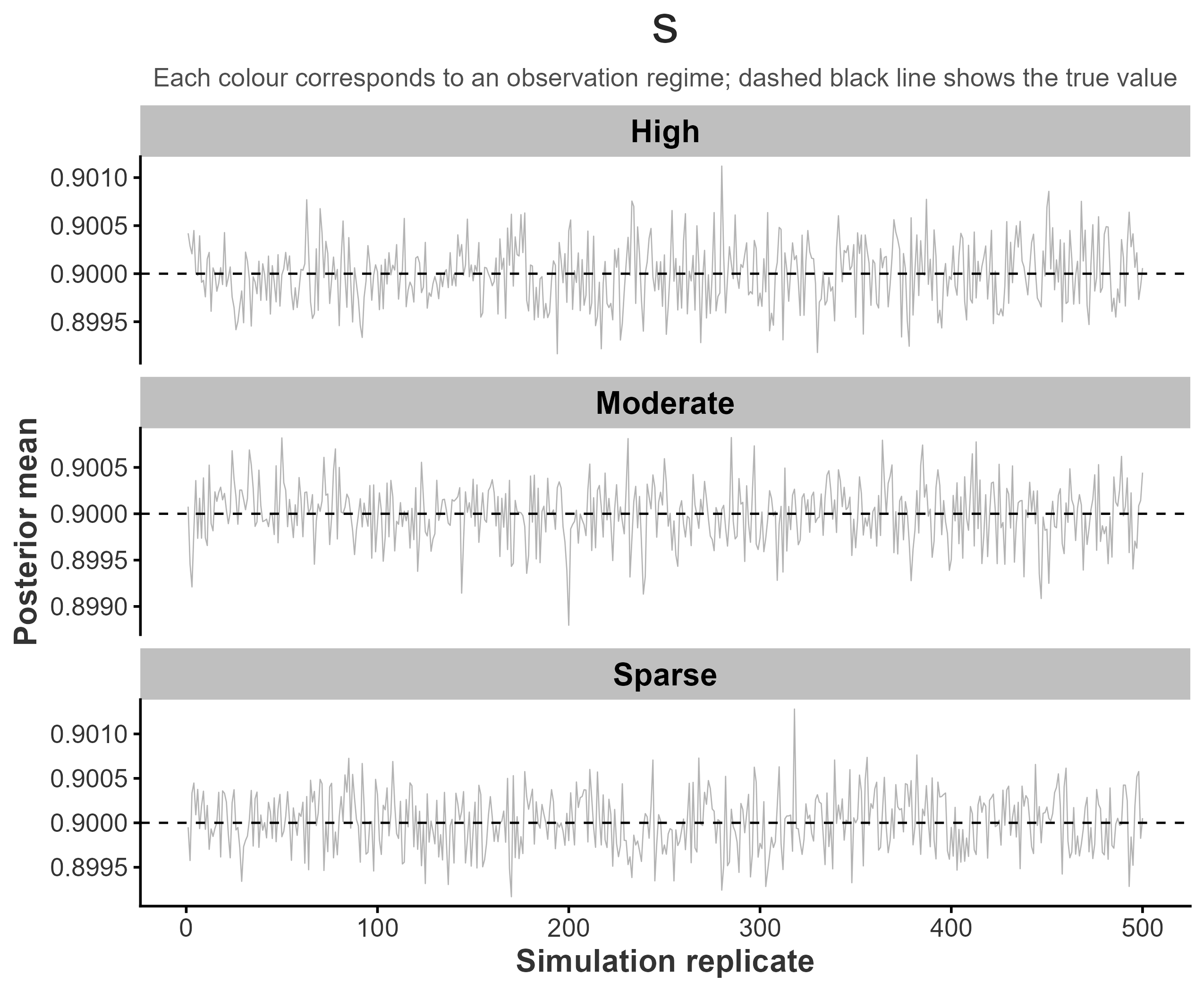}
        \caption*{\small (g) Convergence for $s$}
        \label{fig:conv_s}
    \end{minipage}
    \hfill
    \begin{minipage}[b]{0.47\textwidth}
        \centering
        \includegraphics[width=\textwidth, height=0.20\textheight, keepaspectratio]{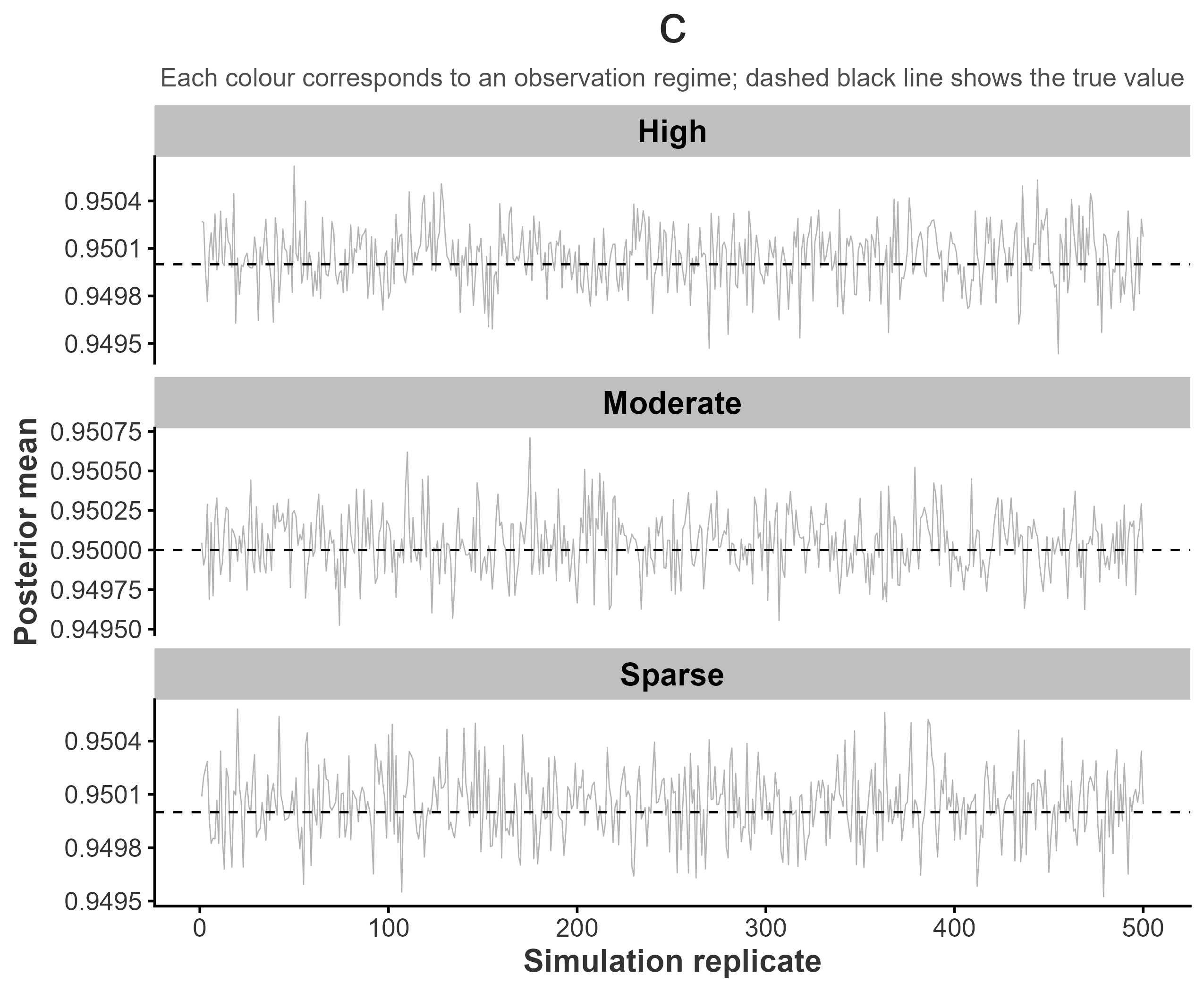}
        \caption*{\small (h) Convergence for $c$}
        \label{fig:conv_c}
    \end{minipage}

    \vspace{0.1cm}
    \caption{MCMC convergence across observation regimes for all monitored parameters ($\beta$, $\xi$, $\kappa$, $\gamma$, $p_E$, $p_I$, $s$, and $c$).}
    \label{fig:mcmc_all_parameters}
\end{figure}

\begin{figure}[htbp]
    \centering
    
    \begin{minipage}[b]{0.48\textwidth}
        \centering
        \includegraphics[width=\textwidth, height=0.5\textheight, keepaspectratio]{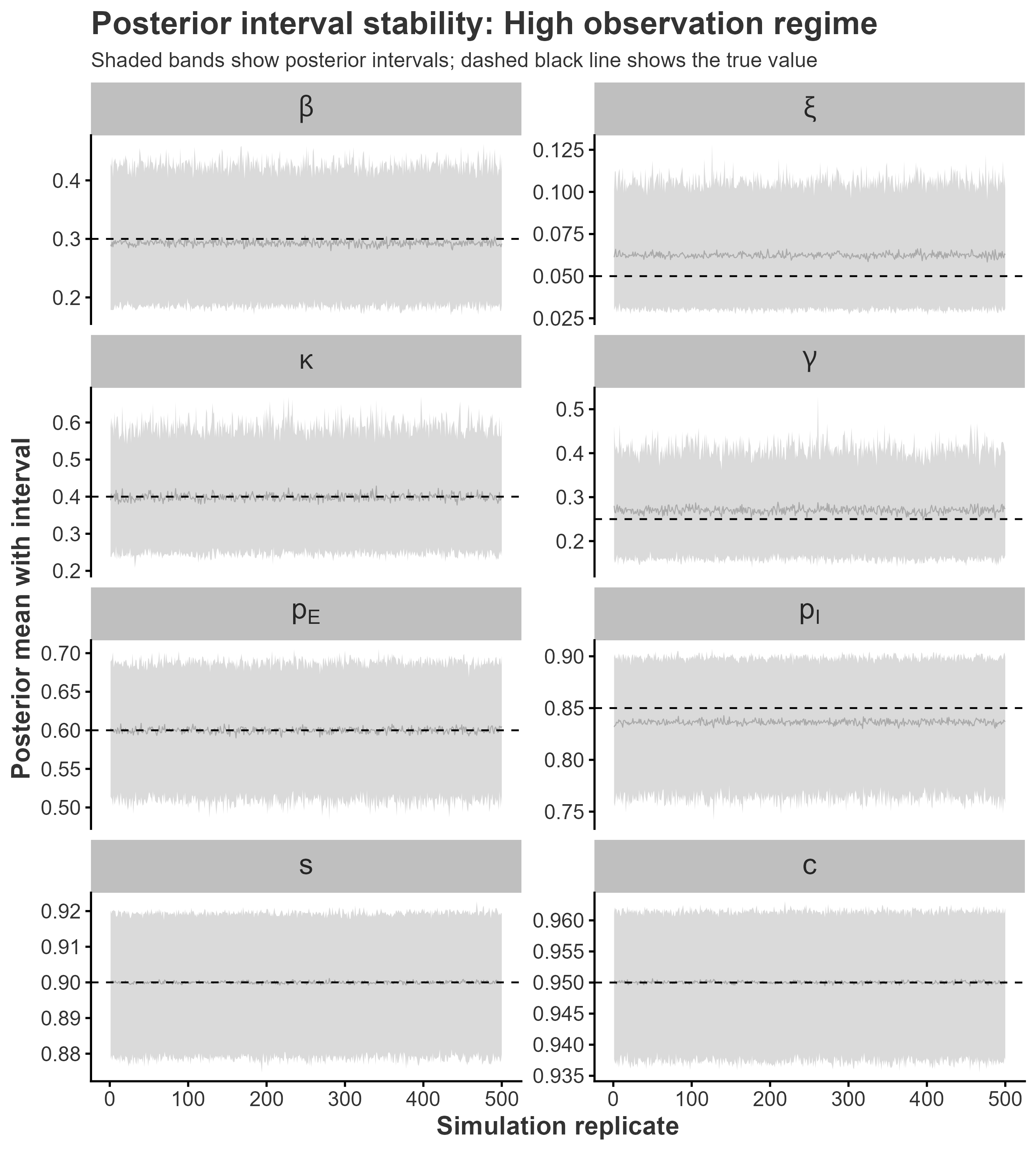}
        \caption*{\small (a) High observation regime.}
        \label{fig:conv_high}
    \end{minipage}
    \hfill 
    \begin{minipage}[b]{0.48\textwidth}
        \centering
        \includegraphics[width=\textwidth, height=0.5\textheight, keepaspectratio]{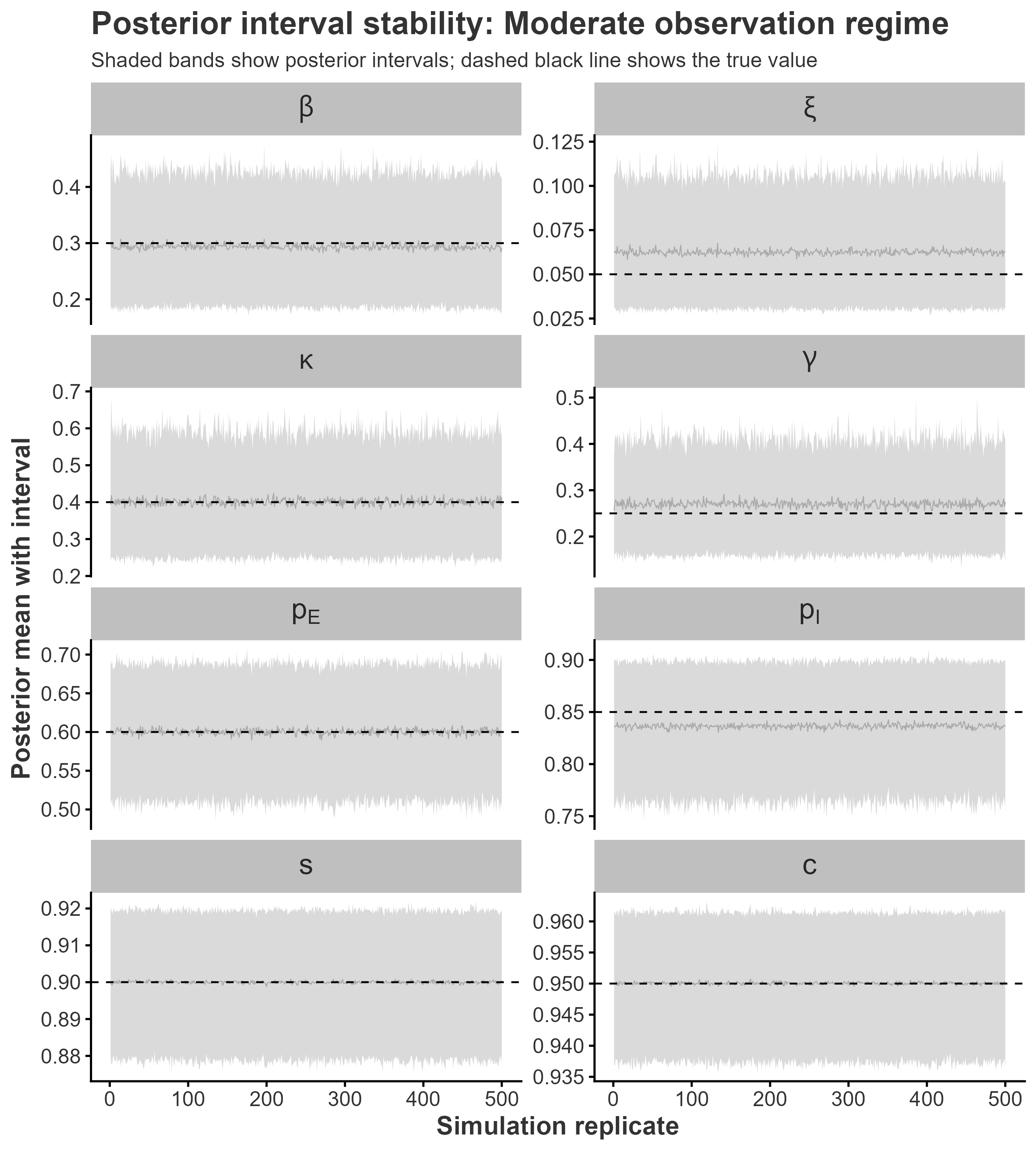}
        \caption*{\small (b) Moderate observation regime.}
        \label{fig:conv_moderate}
    \end{minipage}
    
    \vspace{0.6cm} 

    \begin{minipage}[b]{0.6\textwidth}
        \centering
        \includegraphics[width=\textwidth, height=0.5\textheight, keepaspectratio]{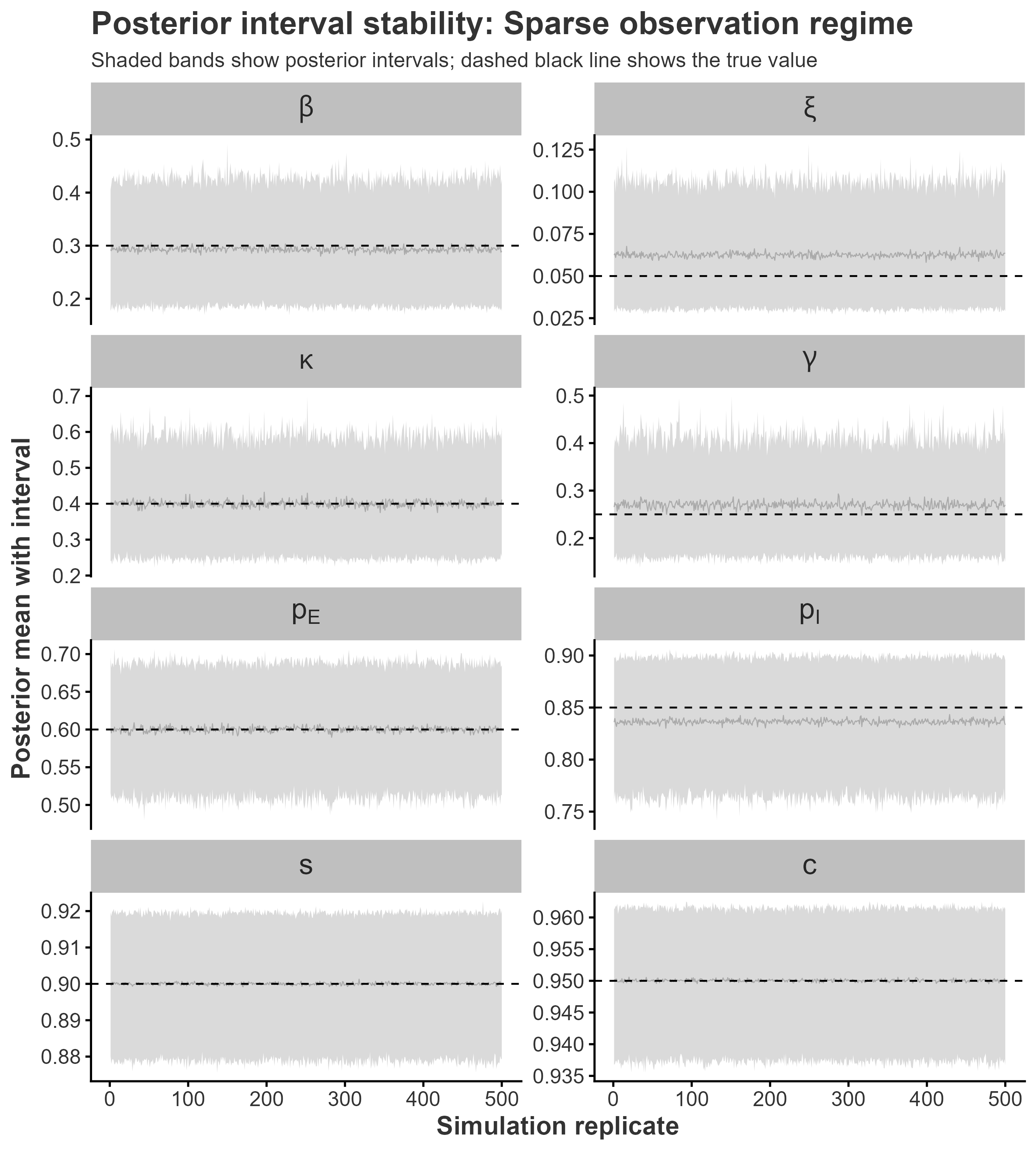}
        \caption*{\small (c) Sparse observation regime.}
        \label{fig:conv_sparse}
    \end{minipage}

    \vspace{0.4cm}
    \caption{Posterior interval stability across replicates for High, Moderate, and Sparse observation regimes.}
    \label{fig:interval_stability_all}
\end{figure}

As expected, the network parameters are more sensitive to missingness than the epidemic parameters. Formation and dissolution rates are estimated well when contact observations are frequent, but accuracy declines as observation intervals widen and misclassification increases. This behavior is visible in the convergence-style plots in Figures~\ref{fig:mcmc_convergence_all} and \ref{fig:interval_stability_all}, which show stable replicate-level posterior means but wider uncertainty bands under the sparser regimes. The increased posterior variability is a sign that the method is appropriately reflecting reduced information rather than overconfidently extrapolating from limited data.

The strongest confounding occurs between the internal infection rate $\beta$ and the external infection rate $\xi$. When contact observations are limited, these two parameters become partially interchangeable from the perspective of the symptom data, leading to broader posterior distributions and mild upward bias in $\xi$ in some settings. This result underscores the importance of network data for separating local transmission from imported infection. By contrast, $\kappa$ and $\gamma$ remain relatively stable across regimes because they are primarily informed by the temporal spacing of observed symptom progression.

Latent trajectory recovery follows the same overall pattern. Under the high-observation regime, the posterior mode of the epidemic history closely matches the true infection and progression times, and the inferred network trajectories capture most major edge changes. Under sparse observation, the latent histories are more uncertain, but the posterior still concentrates around plausible paths consistent with the observed data. This indicates that the model propagates uncertainty appropriately rather than forcing overconfident reconstruction.

\subsection{Sensitivity Analysis}
\label{subsec:sensitivity}

We also examined sensitivity to prior specification and contact misclassification. The posterior estimates were robust to modest changes in the Gamma and Beta hyperparameters, indicating that the data were sufficiently informative in the settings considered. However, when specificity of the contact observations was reduced, false-positive edges led to inflated estimates of network formation rates and, in some cases, exaggerated the apparent opportunity for transmission. This effect was strongest in the sparse-observation regime, where a small number of erroneous contact records could meaningfully alter the inferred latent network structure.

To further assess model behavior, we repeated the simulations under a denser network regime and under weaker transmission intensity. In the denser regime, network parameters were estimated more precisely, but epidemic parameters became slightly harder to separate because more contacts increased the range of plausible transmission pathways. Under weaker transmission, outbreaks were smaller and more variable, which increased uncertainty in all epidemic parameters and reduced the information available for estimating network effects.

\subsection{Summary}
\label{subsec:summary}

Overall, the simulation study indicates that the proposed complete-data likelihood framework performs well when the data contain at least moderate information about both symptoms and contacts. The method is able to recover epidemic, network, and observation parameters with acceptable bias and uncertainty quantification, and it remains usable even when the observed data are incomplete. The results also confirm the central inferential message of the paper: identifiability depends strongly on the richness of the observation regime, especially for parameters governing internal versus external infection and dynamic contact formation.

A particularly important feature of the method is its stability across observation regimes. Even under sparse observation, the posterior summaries remain centered near the truth for most epidemic and observation parameters, and the credible intervals widen in a sensible way to reflect the loss of information. This behavior is exactly what is desired in a Bayesian latent-variable model: the method should not force precision where the data do not support it. The simulation results show that the proposed sampler achieves this balance well.

The results also illustrate the value of joint modeling of epidemic and network dynamics. Parameters that are only weakly identified from symptom data alone, such as the internal and external infection components, become substantially more estimable when contact data are available. Conversely, the network parameters remain the most sensitive to missingness and measurement error, which is unsurprising given their latent nature. These findings underscore the importance of modeling the epidemic and network processes together rather than treating them separately.

In practical terms, the simulation study suggests that the proposed framework is most effective when at least one of the two data sources---symptom information or contact information---is moderately informative, and ideally when both are available. Under those conditions, the model yields accurate parameter recovery, reasonable interval estimation, and stable latent-state reconstruction. This makes the approach well suited for partially observed epidemic systems in which both transmission and contact history matter.

\section{Conclusion}
\label{subsec:conclusion}
This paper develops a Bayesian framework for inference in a dynamic SEIR model on a status-dependent contact network, allowing both the epidemic process and the network evolution to be learned jointly from partially observed data. The simulation study demonstrates that the proposed Metropolis-within-Gibbs sampler can recover the main epidemic, observation, and network parameters accurately and with well-calibrated uncertainty, even when the observation process is sparse. These findings are encouraging because they suggest that the method can be used in settings where direct observation of transmission and contact dynamics is incomplete, which is the norm rather than the exception in many infectious disease applications.

A key strength of the approach is its ability to combine latent-state reconstruction with parameter learning in a single coherent inferential framework. The posterior summaries indicate that the most substantively important parameters are estimated with little bias, while the uncertainty intervals remain appropriately wide for components that are only weakly identified. This balance between precision and honesty is critical in applied work, since overconfident inference on poorly identified latent processes can be misleading. The results here suggest that the model is capable of delivering both accurate point estimates and sensible posterior uncertainty.

The simulation results also highlight the practical usefulness of the regime-based design. By comparing High, Moderate, and Sparse observation settings, we show that the method remains stable as the amount of available information decreases. Although estimation becomes somewhat less precise under sparse observation, the degradation is gradual rather than abrupt, and the posterior intervals continue to exhibit good coverage. This behavior is important for applications in which data quality varies across outbreaks, subpopulations, or surveillance systems.

There are several directions for future work. First, the current implementation can be extended to larger networks and more complex contact structures, including degree heterogeneity, assortative mixing, or time-varying covariates. Second, more efficient proposal mechanisms or blocked updates could further improve mixing for the latent network components. Third, the framework could be generalized to accommodate additional epidemiological states, intervention effects, or multiple pathogen strains. Finally, applying the method to real outbreak data would provide an important assessment of its practical value and computational scalability.

In conclusion, the proposed Bayesian data augmentation strategy provides a flexible and effective tool for inference in dynamic epidemic-network models. The simulation study shows that the method is accurate, stable, and robust across a range of observation regimes, making it a strong candidate for applied inference in partially observed infectious disease systems. The combination of mechanistic modeling, latent process reconstruction, and Bayesian uncertainty quantification offers a powerful framework for studying transmission in dynamic contact networks.






\bibliography{SIERbibliography}

\newpage
\section*{Appendix}
Convergence plots under three data observation regime (high, moderate and sparse) with three different initial guesses (low, medium and high).

\begin{figure}
\includegraphics{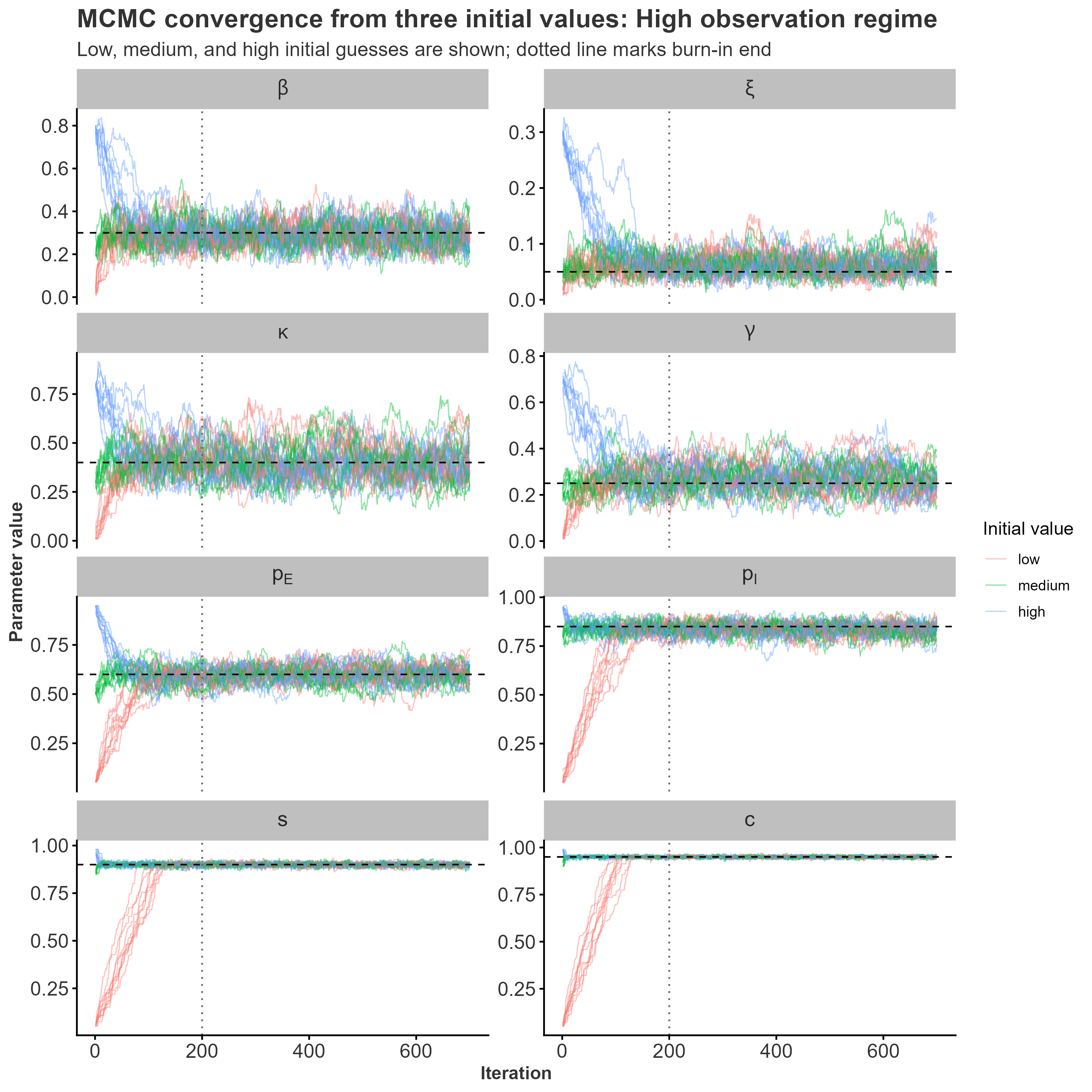}
\end{figure}

\begin{figure}
\includegraphics{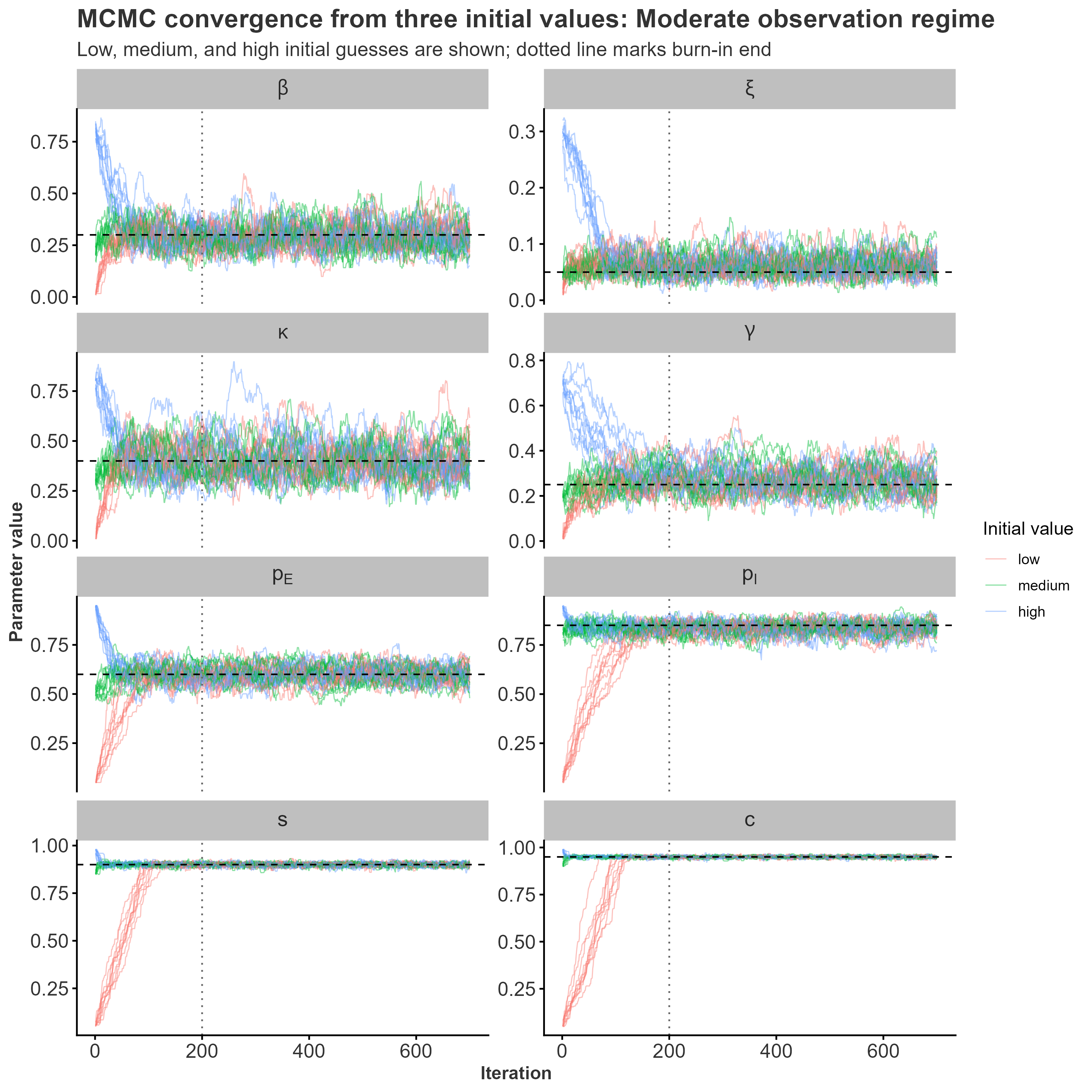}
\end{figure}

\begin{figure}
\includegraphics{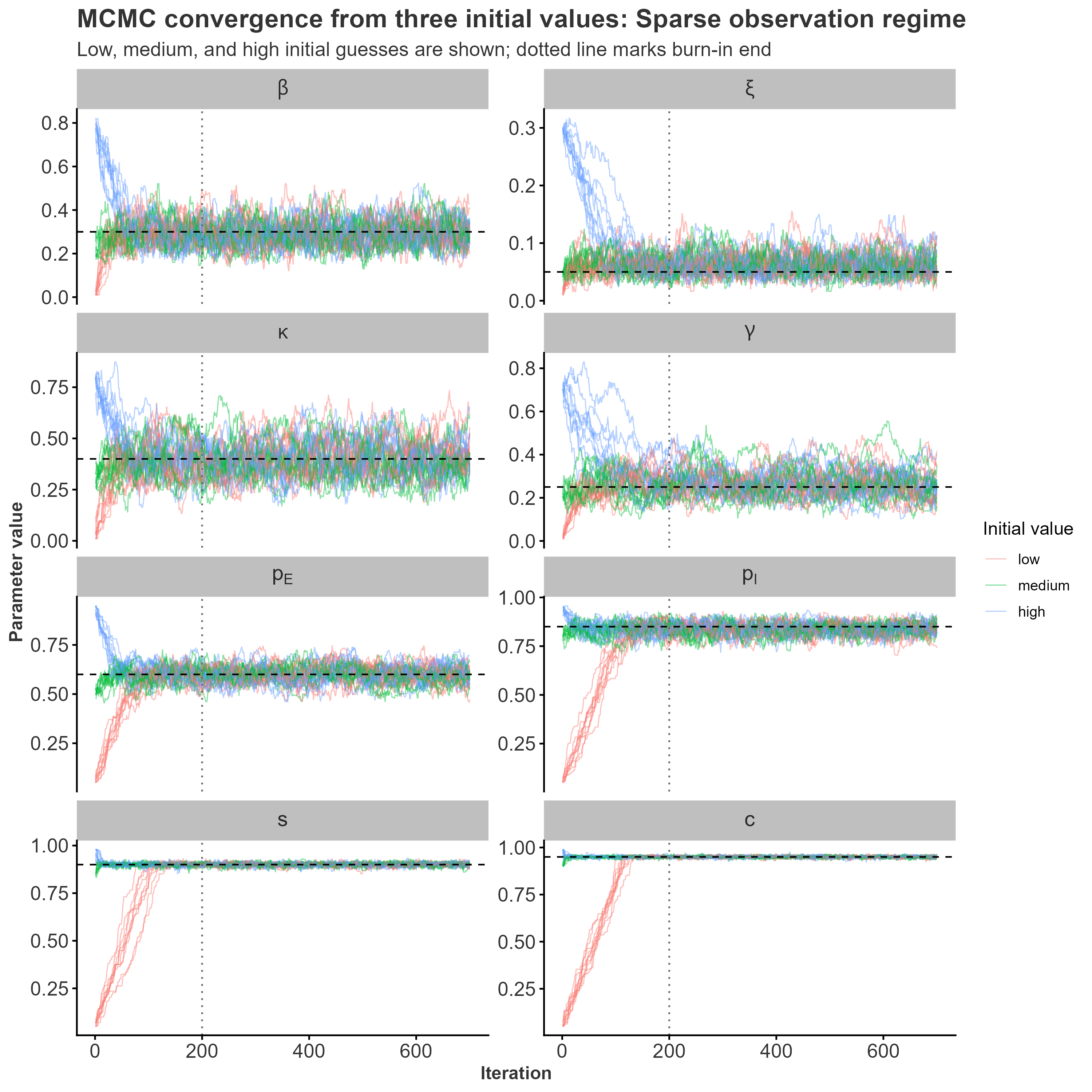}
\end{figure}

\end{document}